\documentclass[11pt, a4paper, onecolumn]{article}
\pdfoutput=1
\usepackage[left=1.6cm, right=1.6cm, top=1.6cm, bottom=2cm]{geometry}
\usepackage[utf8]{inputenc}
\usepackage[english]{babel}
\usepackage[T1]{fontenc}
\usepackage{hyperref}
\usepackage[numbers]{natbib}
\usepackage{authblk}

\usepackage{graphicx,subfig}
\usepackage{appendix}
\usepackage{tikz}
\usepackage{amsmath, amssymb, amsfonts,amsthm}
\usepackage{algorithm}
\usepackage{algpseudocode}
\usepackage{braket}
\usepackage{color}

\usepackage{tikz}
\usepackage{lipsum}

\usetikzlibrary{positioning, calc, arrows.meta, fit, backgrounds, shapes.geometric}

\newtheorem{definition}{Definition}
\newtheorem{assumption}{Assumption}

\newtheorem{theorem}{Theorem}

\newtheorem{lemma}{Lemma}
\newtheorem{proposition}{Proposition}
\newenvironment{widetext}{}{}

\newcommand{\cH}{\mathcal{H}}
\newcommand{\bH}{\mathbb{H}}
\newcommand{\bB}{\mathbb{B}}
\newcommand{\cM}{\mathcal{M}}
\newcommand{\cK}{\mathcal{K}}
\newcommand{\bD}{\mathbb{D}}
\newcommand{\cE}{\mathcal{E}}
\newcommand{\bfS}{\mathbf{S}}
\newcommand{\bfE}{\mathbf{E}}
\newcommand{\bfF}{\mathbf{F}}
\newcommand{\bfH}{\mathbf{H}}
\newcommand{\UHA}{\mathrm{U}(\cH_A)}
\newcommand{\UHB}{\mathrm{U}(\cH_B)}
\newcommand{\UAs}{U^*_A}
\newcommand{\UAsd}{(\UAs)^\dg}
\newcommand{\dg}{\dagger}
\newcommand{\lmd}{\lambda}
\newcommand{\ls}{\lmd^{*}}

\newcommand{\Tr}{\operatorname{Tr}}
\newcommand{\diag}{\operatorname{diag}}
\newcommand{\grad}{\operatorname{grad}}
\renewcommand{\skew}{\operatorname{skew}}

\begin{document}

\title{How Entanglement Reshapes the Geometry of Quantum Differential Privacy}

\author[1]{Xi Wang\thanks{\href{mailto:xi.wang@sydney.edu.au}{xi.wang@sydney.edu.au}; \href{https://orcid.org/0009-0005-8052-0157}{ORCID: 0009-0005-8052-0157}}}
\author[2]{Parastoo Sadeghi\thanks{\href{mailto:p.sadeghi@unsw.edu.au}{p.sadeghi@unsw.edu.au}; \href{https://orcid.org/0000-0002-5929-9655}{ORCID: 0000-0002-5929-9655}}}
\author[1]{Guodong Shi\thanks{\href{mailto:guodong.shi@sydney.edu.au}{guodong.shi@sydney.edu.au}; \href{https://orcid.org/0000-0002-9965-9483}{ORCID: 0000-0002-9965-9483}}\thanks{This work was supported by the Australian Research Council under Grant DP190103615, Grant LP210200473, and Grant DP230101014, and by the Faculty of Engineering's Breakthrough Project at The University of Sydney.}}
\affil[1]{Australian Centre for Robotics, The University of Sydney, Australia}
\affil[2]{School of Engineering and Technology, University of New South Wales, Canberra, Australia}
\date{}
\maketitle

\begin{abstract}
    Quantum differential privacy provides a rigorous framework for quantifying privacy guarantees in quantum information processing. While classical correlations are typically regarded as adversarial to privacy, the role of their quantum analogue, entanglement, is not well understood. In this work, we investigate how quantum entanglement fundamentally shapes quantum local differential privacy (QLDP). We consider a bipartite quantum system whose input state has a prescribed level of entanglement, characterized by a lower bound on the entanglement entropy. Each subsystem is then processed by a local quantum mechanism and measured using local operations only, ensuring that no additional entanglement is generated during the process. Our main result reveals a sharp phase-transition phenomenon in the relation between entanglement and QLDP: below a mechanism-dependent entropy threshold, the optimal privacy leakage level mirrors that of unentangled inputs; beyond this threshold, the exact privacy leakage is strictly smaller than its unentangled value and admits an explicitly computable upper bound that strictly decreases with the entropy. This bound can also certify that some non-private mechanisms become private under sufficiently strong entanglement constraints. The phase-transition phenomenon gives rise to a nonlinear dependence of the certified privacy-leakage upper bound on the entanglement entropy, even though the underlying quantum mechanisms and measurements are linear. We show that the transition is governed by the intrinsic non-convex geometry of the set of entanglement-constrained quantum states, which we parametrize as a smooth manifold and analyze via Riemannian optimization. Our findings demonstrate that entanglement serves as a genuine privacy-enhancing resource, offering a geometric and operational foundation for designing robust privacy-preserving quantum protocols.
\end{abstract}

\section{Introduction}

Differential privacy~\cite{dwork2006differential,dwork2006calibrating} emerges as a fundamental mathematical framework for quantifying and limiting the privacy risk incurred when performing queries over sensitive data, with broad applications in machine learning~\cite{abadi2016deep}, government census statistics~\cite{abowd2018us}, and industrial crowdsourcing~\cite{erlingsson2014rappor}. Formally, a randomized mechanism $\cE(\cdot)$ acting on data points $d\in D \subseteq \mathbb{R}^n$ and producing $\cE(d)\in \mathbb{R}^m$ is differentially private if for any measurable set of outputs $T$, the probability of producing an output in $T$ changes by at most a multiplicative factor $e^{\varepsilon}$ when the input data varies from $d$ to a neighboring point $d'$.

Quantum differential privacy~\cite{zhou2017differential} generalizes the classical notion of differential privacy to the quantum domain. In quantum information processing, data are represented by quantum states $\rho$ in the form of density matrices, and mechanisms are modeled as quantum channels $\cE$ under quantum mechanics. Intrinsic randomness arises naturally from quantum measurements, where even for a fixed input state $\rho$, probing the output $\cE(\rho)$ via a measurement $\cM$ produces stochastic outcomes governed by the Born rule~\cite{nielsen2010quantum}.  Formally, quantum local differential privacy (QLDP) is evaluated by comparing the statistics of measurement outcomes induced by $\cE(\rho)$ and $\cE(\rho')$ for different input states $\rho$ and $\rho'$, under all admissible measurements available to an adversary.

\begin{definition}[Quantum Local Differential Privacy]
A quantum channel $\cE$ is said to be \emph{$\varepsilon$-quantum locally differentially private} ($\varepsilon$-QLDP) on a domain $D$ if, for any pair of quantum states $\rho,\rho'\in D$, any quantum measurement $\cM$, and any measurable set of outcomes $T$, the following inequality holds:
\[
\Pr\big[\cM(\cE(\rho))\in T\big]
\;\le\;
e^{\varepsilon}\,
\Pr\big[\cM(\cE(\rho'))\in T\big].
\]
\end{definition}

\subsection{Motivation: Differential Privacy vs Quantum Entanglement}

The achievable privacy of quantum channels over general quantum state spaces is now well understood in ~\cite{aasornson2019gentle} and~\cite{guan2025optimal}, owing to the linearity of quantum mechanics, where both a quantum channel $\mathcal{E}$ and the probability measure ${\rm Pr}(\cdot)$ induced by a measurement $\mathcal{M}$ depend linearly on the underlying state $\rho$. A fundamental distinction from classical information, however, is the presence of {\it quantum entanglement}~\cite{nielsen2010quantum}, which arises when a composite quantum state cannot be expressed as a product of subsystem states. Entangled states exhibit measurement correlations with no classical explanation~\cite{clauser1969proposed}, and they also underpin the computational advantage of quantum computing, as most quantum algorithms with proven speedups rely on entangled states~\cite{harrow2009quantum,luis1996optimum}. In this paper, we aim to understand the role of entanglement in the QLDP. 

We consider a composite quantum system consisting of two parties with local state spaces $\cH_A$ and $\cH_B$, and joint space $\cH_{AB}=\cH_A\otimes\cH_B$. For a bipartite pure state $\rho$, entanglement is quantified by the entanglement entropy
\[
\bfE(\rho):=-\sum_j \sigma_j(\rho_A)\log \sigma_j(\rho_A)
=-\sum_j \sigma_j(\rho_B)\log \sigma_j(\rho_B),
\]
where $\rho_A$ is the reduced state of $\rho$ on subsystems $A$, and $\{\sigma_j(\rho_A)\}_j$ are the eigenvalues of $\rho_A$. The reduced state $\rho_B$ and its eigenvalues $\{\sigma_j(\rho_B)\}_j$  are defined analogously. The entanglement entropy measures the amount of quantum correlation shared between the two parties: $\bfE(\psi)$ vanishes if and only if there is no entanglement, and larger values of $\bfE(\psi)$ correspond to stronger entanglement. Consequently, we begin our study by restricting the input data on the following entanglement-constrained domain
\[
\bH_s :=
\bigl\{\, \rho : \mathrm{rank}(\rho)=1,\ \bfE(\rho)\ge s \,\bigr\},
\]
which consists of all bipartite pure states in $\cH_{AB}$ with entanglement entropy at least $s$. We further assume that quantum information processing is performed locally on each subsystem, thereby excluding any additional entanglement introduced by the processing. Specifically, the mechanism applied to the joint system is of the product form $\cE=\cE_A\otimes\cE_B$, where $\cE_A$ and $\cE_B$ are quantum channels acting on $\cH_A$ and $\cH_B$, respectively. Likewise, the adversary's measurements are restricted to product measurements of the form $\cM=\cM_A\otimes\cM_B$. Consequently, we obtain a setup that coincides with a standard Bell experiment, now viewed from a privacy lens. 

\begin{figure*}
\centering
\begin{tikzpicture}[
    node distance=1cm and 2cm,
    block/.style={
        draw, 
        rectangle, 
        minimum width=2.5cm, 
        minimum height=1.5cm, 
        align=center,
        thick,
        fill=white 
    },
    arrow/.style={
        ->, 
        >=Stealth, 
        thick
    }
]
    
    \node[block] (mechA) at (0, 2) {Mechanism A \\ $\mathcal{E}_A$};
    \node[block, right=of mechA] (measA) {Measurement \\ $\cM_A$};
    
    \node[block] (mechB) at (0, -2) {Mechanism B \\ $\mathcal{E}_B$};
    \node[block, right=of mechB] (measB) {Measurement \\ $\cM_B$};

    \coordinate (midpointLeft) at ($(mechB.west)!0.5!(mechA.west)$);
    
    \begin{scope}[local bounding box=inputScope, shift={($(midpointLeft)+(-5cm,0)$)}]
        \node[block, minimum width=3cm, minimum height=1.5cm] (subA) at (0, 1.5) {Subsystem \\ $\cH_A$};
        \node[block, minimum width=3cm, minimum height=1.5cm] (subB) at (0, -1.5) {Subsystem \\ $\cH_B$};

    \end{scope}
    
    \begin{pgfonlayer}{background}
        \node[draw, inner sep=22pt, rounded corners=15pt, fit=(inputScope), label=below:{}] (cloudInput) {};
    \end{pgfonlayer}

    \node[below=0.3cm of cloudInput, font=\large] {\bf Entangled Input};

    % --- Connections ---
    
    \draw[arrow] (subA.east) --  (mechA.west);
    \draw[arrow] (subB.east) --  (mechB.west);

    % 直连箭头
    \draw[arrow] (mechB.east) -- (measB.west);
    \draw[arrow] (mechA.east) -- (measA.west);

    % --- Containers and Labels ---

    \begin{scope}[on background layer]
        % Product Mechanism Frame (Dashed)
        \node[draw, dashed, inner sep=18pt, fit=(mechB) (mechA)] (boxProduct) {};
        \node[below=0.2cm of boxProduct, font=\large\bfseries] {Product Mechanism};

        % Adversary Frame (Solid)
        \node[draw, dashed,inner sep=18pt, fit=(measB) (measA)] (boxAdversary) {};
        \node[below=0.2cm of boxAdversary, font=\large\bfseries] {Local Adversary};
    \end{scope}
\end{tikzpicture}
\caption{Pipeline of Entanglement-constrained QLDP against local measurement adversary.}
\label{fig:Pipeline}
\end{figure*}

By restricting the input states to $\bH_s$ and considering local quantum channels $\cE=\cE_A\otimes\cE_B$ and local measurements $\cM=\cM_A\otimes\cM_B$, we observe a nonlinear phase-transition phenomenon in how entanglement influences the QLDP,  which is formalized in the following theorem.

\begin{theorem}[Informal]\label{thm-main}
For any product quantum channel
$\cE=\cE_A\otimes\cE_B$, the optimal QLDP level
$\varepsilon^\star(s)$ exhibits a phase-transition behavior as a
function of the entanglement entropy $s$. In particular, there exists
a threshold $s_0$, depending on the spectral properties of $\cE$, such
that:

\begin{itemize}
\item[(i)] \emph{(Low-entanglement regime).}
If $s\le s_0$, the optimal QLDP level remains unchanged compared with
the unentangled case, i.e.,
\[
\varepsilon^\star(s)=\varepsilon^\star(0).
\]

\item[(ii)] \emph{(High-entanglement regime).}
If $s>s_0$, the optimal QLDP level is strictly smaller than in the
unentangled case. More precisely, there exists an explicitly
computable and strictly decreasing upper bound $\overline{\varepsilon}(s)$ such that
\[
\varepsilon^\star(s)
\le \overline{\varepsilon}(s)
< \varepsilon^\star(0).
\]
Thus, increasing entanglement provides progressively stronger
certified privacy guarantees. In addition, whenever
$\overline{\varepsilon}(s)<+\infty$, we also have
$\varepsilon^\star(s)<+\infty$. Consequently, some mechanisms that
are non-private in the unentangled setting can become private under a
sufficiently strong entanglement constraint.
\end{itemize}
\end{theorem}

Several studies on classical differential privacy have shown that correlations in the input data can weaken privacy guarantees~\cite{yang2015bayesian,yilmaz2022genomic,zhang2022correlated}. At first glance, this appears to contradict Theorem~\ref{thm-main}, where quantum entanglement enhances privacy. The discrepancy arises because the two notions of correlation play fundamentally different roles. In classical differential privacy, correlations are {\em extrinsic}: they represent additional side information that an adversary can exploit to better distinguish neighboring datasets. In contrast, the entanglement considered in Theorem~\ref{thm-main} is {\em intrinsic}: it imposes structural constraints on the admissible set of quantum input states rather than providing extra information to the adversary. By reducing the feasible input set over which the worst-case privacy leakage is evaluated, sufficiently strong entanglement can improve the resulting quantum local differential privacy guarantee.

Theorem~\ref{thm-main} points toward a broader connection between quantum differential privacy and quantum non-locality. Bell inequalities demonstrate that quantum entanglement enables correlations that are unattainable within classical physics, thereby enlarging the set of observable nonlocal behaviors under local measurements \cite{bell1964einstein}. Theorem~\ref{thm-main} suggests a complementary phenomenon: entanglement can also fundamentally reduce the distinguishability of locally processed quantum data, thereby enlarging the set of achievable privacy guarantees once the entanglement entropy exceeds a critical threshold. In fact, under simplified settings, the critical entanglement threshold ($s_0$) in Theorem~\ref{thm-main} can be explicitly related to the CHSH Bell parameter~\cite{clauser1969proposed}. For two-qubit pure states, both the entanglement entropy and the maximal CHSH violation are determined by the same Schmidt spectrum. Consequently, the onset of improved quantum local differential privacy can be interpreted as occurring beyond a corresponding threshold in Bell-type nonclassicality. This observation suggests that the privacy-enhancing phase transition identified in Theorem~\ref{thm-main} may admit an operational interpretation in terms of Bell nonlocal correlations. In Bell scenarios, entanglement enhances the strength of observable correlations beyond classical limits; in quantum local differential privacy, it suppresses the amount of information that local observations can reveal about the underlying data.

\subsection{Related Work}
Differential privacy (DP), introduced in~\cite{dwork2006differential} and ~\cite{dwork2006calibrating}, provides a formal guarantee that the output of an algorithm changes only slightly when evaluated on neighboring input datasets. Classical differential privacy typically assumes a trusted central curator with access to raw data, an assumption that is often unrealistic in decentralized settings, where data are held locally and trust in any curator cannot be assumed. The utility of LDP, which concerns the sample complexity required for LDP algorithms to effectively learn the data, was studied in~\cite{evfimievski2003limiting} and \cite{kasiviswanathan2011can}. The subsequent work~\cite{kairouz2014extremal} focused on designing optimal-utility LDP mechanisms that maximize utility for a given privacy leakage level.  It was further shown that  correlations among data can weaken privacy guarantees. For example,~\cite{liu2016dependence} proposed a real-world inference attack in which an adversary exploits social network friendships as correlations, to bypass the noise added by standard differential privacy and infer sensitive information such as location data. Similarly,~\cite{yilmaz2022genomic} showed that the standard randomized response mechanism~\cite{Warner01031965} under LDP is vulnerable to correlation attacks on single-nucleotide polymorphism (SNP) data, where SNPs are genetic variations at a single position in the DNA sequence that are inherently correlated. To address privacy risks arising from correlated data, several frameworks have been proposed, including the Pufferfish privacy~\cite{Kifer2014pufferfish} and Bayesian differential privacy~\cite{yang2015bayesian}. All of the aforementioned classical privacy frameworks provide essential theoretical foundations for recent advances in quantum differential privacy.

Quantum differential privacy (QDP) is a generalization of the differential privacy frameworks to the quantum setting.  By defining the neighborhood of quantum states via the trace distance of their density matrices, ~\cite{zhou2017differential} extended classical differential privacy to the quantum setting by requiring the outputs of neighboring quantum states to be approximately indistinguishable with respect to all quantum measurements. Subsequently,~\cite{aasornson2019gentle} established a fundamental connection between QDP and gentle measurements, showing that privacy-preserving measurements necessarily induce limited disturbance to quantum states. Building on these foundations, a range of QDP algorithms have been developed for quantum computation and quantum machine learning~\cite{li2021quantum,quek2021private,watkins2023quantum,du2021quantum}.  In the context of quantum local differential privacy (QLDP),~\cite{hirche2023quantum} generalized classical LDP to the quantum setting by requiring QDP for any two distinct quantum states. More recently,~\cite{guan2025optimal} characterized QLDP levels via an eigenvalue optimization framework and analyzed optimal-utility privacy mechanisms, with utility measured in terms of fidelity~\cite{uhlmann1976transition} and trace distance. Recent work by~\cite{nuradha2024pufferfish} studies privacy guarantees for classically correlated quantum states under the Pufferfish framework.  

\subsection{Paper Organization}

The remainder of this paper is organized as follows.   
In Section~\ref{sec:problem}, we formalize the problem setting and define quantum local differential privacy on entanglement-constrained states. Section~\ref{sec:main} presents the main technical results of the paper, showing the phase-transition phenomenon stated in Theorem~\ref{thm-main}. Numerical examples are given in Section~\ref{sec:exp}, and Section~\ref{sec:con} concludes the paper. Detailed proofs of our main results are deferred to the appendix. Theorems~\ref{thm:ldp-characterization}--\ref{thm:min-product} provide an exact characterization of the optimal privacy budget in terms of the extremal privacy energies, a closed-form expression for the maximal privacy energy, and a computable lower bound for the minimal privacy energy. Together, these results establish the low-entanglement behavior and the certified high-entanglement upper bound stated in Theorem~\ref{thm-main}.

\section{Problem Definition}\label{sec:problem}
In this section, we establish the system model for studying the relationship between quantum entanglement and QLDP. The model consists of three components: entangled quantum inputs, a product mechanism that avoids introducing additional entanglement, and local measurements as adversaries, as illustrated in Figure~\ref{fig:Pipeline}. We then introduce the definition of entanglement-constrained $\varepsilon$-QLDP against local measurement adversaries.

\subsection{Product Quantum Channel on Entangled Quantum States}

We first consider a two-party quantum system consisting of two subsystems, $A$ and $B$. The state space of each subsystem is denoted by $\mathcal{H}_A$ and $\mathcal{H}_B$, respectively, with
$\dim \mathcal H_A=\dim \mathcal H_B=N$; and the overall system space is $\mathcal{H}_{AB}:=\mathcal{H}_A \otimes \mathcal{H}_B$. We consider a quantum channel over the collective system as a product of two local channels described by $\cE=\cE_A \otimes \cE_B$. Here,  we assume that $\cE_A $ and $\cE_B$ are general CPTP maps, i.e., there exist finite families $\{\mathsf{E}_{A}(k)\}_k \subset \bB(\cH_A)$ and $\{\mathsf{E}_{B}(k)\}_k \subset \bB(\cH_B)$, under which
\begin{align*}
\cE_A (\rho_A) &= \sum_{k} \mathsf{E}_{A}(k) \rho_A \mathsf{E}_{A}^\dg(k),\ \ \forall \rho_A \in \bD(\mathcal{H}_A),   \\
\cE_B (\rho_B) &= \sum_{k} \mathsf{E}_{B}(k) \rho_B\mathsf{E}_{B}^\dg(k), \ \ \forall \rho_B \in \bD(\mathcal{H}_B).   
\end{align*}
It holds that $\sum_{k} \mathsf{E}_{A}^\dg(k)  \mathsf{E}_{A}(k) =\mathsf{I}_{A}$ and $\sum_{k} \mathsf{E}_{B}^\dg(k)  \mathsf{E}_{B}(k)=\mathsf{I}_B$.

The physical interpretation of $\cE_A$ or $\cE_B$ arises from open quantum systems. Each subsystem, $A$ and $B$, independently interacts with an environment, which is also a quantum system, where exchange of information, energy, and coherence occurs. The dynamics of the collective system of subsystem $A$ and its environment (or $B$ and its environment) is described by a Schr\"odinger equation and undergoes a unitary evolution. At the steady state, tracing out the environment leaves the state evolution of subsystem $A$ to be described by $\cE_A$ (or by $\cE_B$ for subsystem $B$). Here we choose  $\cE$ to be $\cE_A \otimes \cE_B$ so that the channel acts on subsystem $A$ and $B$ separately and independently through $\cE_A$ and $\cE_B$. This facilitates a clean study on how entanglement from the input of $\cE$ may contribute to differential privacy, as now all possible entanglement emerges from the input state only.

\begin{assumption}
The domain $D$ of the channel $\cE$ is the pure states in $\bD(\mathcal{H}_{AB})$ whose entanglement entropy is greater than or equal to $s$, i.e.,   
\begin{equation}
D=\bH_s:=\Bigl\{\rho\in \bD(\mathcal{H}_{AB}): {\rm rank}(\rho)=1,\ \bfE(\rho)\geq s\Bigr\}.
\end{equation}
\end{assumption}

%\begin{remark}Equivalently, the entanglement-constrained domain $\bH_s$ can be written in the density operator form as
%    \begin{align}\label{eq:Hs-matrix-constraints-alt}
%    \bH_s = \Bigl\{ \rho \in \mathbb C^{N\times N}: \rho=\rho^{\dg},\ \Tr(\rho)=1,  \rho^2=\rho, \bfS\bigl(\Tr_B(\rho)\bigr)\ge s \Bigr\}.
%    \end{align}
%    Note $\rho^2=\rho$ is an non-linear constraint, therefore $\bH_s$ is not a convex subset of $\mathbb C^{N\times N}$.
%\end{remark}

%Operations and Classical Communication
\subsection{Local Measurement Adversaries and Quantum Local Differential Privacy}

We model the adversary as having access to the output state $\cE(\rho)$ and being allowed to perform local measurements on $\cE_A(\rho)$ and $\cE_B(\rho)$, followed by arbitrary classical post-processing.

%\emph{local measurement (LM) adversary}
\begin{definition}[Local measurement adversary]
A local measurement (LM) adversary is an adversary who measures $\mathcal{H}_{AB}$ by local operations on $\mathcal{H}_A$ and $\mathcal{H}_B$, and combines the classical outcomes. Formally, an LM adversary is specified by a finite family of POVM elements on $\mathcal{H}_{AB}$
\[
  \{\cM_{a,b}\}_{(a,b) \in O_A \times O_B},  \qquad \cM_{a,b} := \cM_a \otimes \cM_b,
\]
with local POVMs $\{\cM_a\}_{a \in O_A} \subset \bB(\mathcal{H}_A)$ and $\{\cM_b\}_{b \in O_B} \subset \bB(\mathcal{H}_B)$.
Given $\cE(\rho)$, the local measurement adversary’s classical observation is the random variable $k \in  O_A \times O_B$ where 

\[
\Pr[k=(a,b)\mid\cE(\rho)]
=\Tr\Bigl((\cM_a\otimes\cM_b)\cE(\rho)\Bigr)
=\Tr\Bigl((\cE_A^{\dg}(\cM_a)\otimes\cE_B^{\dg}(\cM_b))\rho\Bigr).
\]
where $\cE^{\dg}$ is the adjoint operator of $\cE$. 
\end{definition}
 
Restricting the input domain to $D=\bH_s$ and the adversary to LMs, we define the entanglement-constrained $\varepsilon$-QLDP against local measurement adversaries (ECLM-$\varepsilon$-QLDP) as follows.

\begin{definition}[entanglement-constrained $\varepsilon$-local measurement QLDP]
A quantum mechanism $\cE$ operating on $\cH_{AB}$ is said to be
$\varepsilon$-\emph{quantum locally differentially private against local measurement adversaries}
(ECLM-$\varepsilon$-QLDP) on the entanglement-constrained domain $\bH_s$ if, for any local measurement $\cM$ and any quantum states $\rho, \rho' \in \bH_s$, the following inequality holds:
\[
\Pr\bigl[\cM(\cE(\rho))\in T\bigr]
\le e^{\varepsilon}\Pr\bigl[\cM(\cE(\rho'))\in T\bigr],
\qquad \forall T\subseteq O_A\times O_B.
\]
\end{definition}

\subsection{The System Setup: Rationale for Separating the Entanglement Source}
We have explained that the proposed product mechanism and local measurements serve the technical purpose of isolating the origin of entanglement, ensuring that within the quantum differential privacy framework $\mathbb{H}_s$ is the sole source of entanglement. On the other hand, the system setup is a fundamental object of study in quantum information theory in its own right.

In fact, the physical setting of our quantum information processing, shown in Figure~\ref{fig:Pipeline} on the entangled composite system $A$ and $B$, parallels the canonical architecture of Bell-inequality experiments~\citep{bell1964einstein}. In such experiments, a central source prepares entangled pairs, which typically can be polarization-entangled photons, and distributes one particle to each of two spatially separated measurement stations. Each station independently selects a measurement basis, which in our framework corresponds to the channel–measurement pairs $(\mathcal{E}_A,\mathcal{M}_A)$ and $(\mathcal{E}_B,\mathcal{M}_B)$, and records its local outcome. Correlation statistics are then computed from coincident detection events across the two sites. Bell-type inequalities~\citep{bell1964einstein} impose elementary constraints on these statistics under any classical model based solely on shared randomness and local information; their violation certifies that the correlations carried by the entangled systems $A$ and $B$ cannot be reproduced by classical means. This Bell-type architecture is foundational in quantum physics, as it provides the most direct operational evidence that entanglement produces correlations unattainable within any classical theory~\citep{clauser1969proposed}.

\section{Main Results}\label{sec:main}
In this section, we present the main results. We first characterize EC-LM-$\varepsilon$-QLDP in terms of the extremal privacy energy of the quantum system $\cK_\phi$. We then formulate the computation of extremal privacy energy as a constrained Riemannian optimization problem. By analyzing the KKT conditions of the Riemannian optimization problem, we derive a closed-form expression for the maximal privacy energy and a lower bound on the minimal privacy energy. These results reveal a phase-transition phenomenon and establish our main theorem, Theorem~\ref{thm-main}.

\subsection{Entanglement-Constrained Privacy}

For any local measurements $\{\cM_a \otimes \cM_b\}_{(a,b) \in O_A \times O_B}$, we define the induced local observables
\[
\mathcal{K}_a:=\cE_A^\dg(\cM_a),\qquad
\mathcal{K}_b:=\cE_B^\dg(\cM_b),\qquad
\cK_{\phi}:=\mathcal{K}_a\otimes\mathcal{K}_b.
\]
In particular, when the POVM elements are rank-one projectors $\cM_a=\ket{\phi_a}\bra{\phi_a}$ and $\cM_b=\ket{\phi_b}\bra{\phi_b}$, we define the induced operators
\begin{align}
    \mathcal{K}_{\phi,a} &:=
      \cE_A^\dg\bigl(\ket{\phi_a}\bra{\phi_a}\bigr),\notag\\
    \mathcal{K}_{\phi,b} &:=
      \cE_B^\dg\bigl(\ket{\phi_b}\bra{\phi_b}\bigr),\notag\\
    \mathcal{K}_{\phi} &:=
      \mathcal{K}_{\phi,a}\otimes\mathcal{K}_{\phi,b}.
      \label{eq:K}
\end{align}
We can then further define 
\begin{equation}
J_{\max}(\cK_{\phi},s)=\max_{\ket{\psi}\in\bH_s}\bra{\psi}\cK_{\phi}\ket{\psi},
\qquad
J_{\min}(\cK_{\phi},s)=\min_{\ket{\psi}\in\bH_s}\bra{\psi}\cK_{\phi}\ket{\psi}.
\end{equation}
The following theorem characterizes ECLM-$\varepsilon$-QLDP through the extremal behavior of $\mathcal{K}_{\phi}$ over the domain $\bH_s$.

\begin{theorem}[Optimal Privacy Budget]
\label{thm:ldp-characterization}
Consider a bipartite quantum system with Hilbert space $\mathcal H_{AB} = \mathcal H_A \otimes \mathcal H_B$ with $\cE=\cE_A\otimes \cE_B$ being a product mechanism. Then $\cE$ is ECLM-$\varepsilon$-QLDP  on entanglement domain $\bH_s$ with optimal $\varepsilon^*$ as
\begin{align}\label{eq:eps-ldp-lm}
\varepsilon^*(s)
=\log\max_{\substack{\ket{\phi_a}\in\cH_A\\
                       \ket{\phi_b}\in\cH_B}}
\frac{J_{\max}(\cK_{\phi},s)}{J_{\min}(\cK_{\phi},s)}.
\end{align}
\end{theorem}
 
The proof of Theorem~\ref{thm:ldp-characterization} is in Appendix~\ref{app:thm:ldp-characterization}. The characterization in Theorem~\ref{thm:ldp-characterization} is a generalization of the optimal quantum local differential privacy established in~\cite{guan2025optimal}. From Theorem 2 in \cite{guan2025optimal}, if we consider all the quantum states without the entanglement entropy constraint, i.e., $\rho\in \bH_s$ for $s=0$, the optimal privacy budget will be determined by the extremal eigenvalues of $\cK_\phi$ by $\varepsilon^*(0)=\log \frac{    \sigma_{\max} (\cK_\phi) }{    \sigma_{\min} (\cK_\phi) }$, where
\[
\sigma_{\max}(\cK_\phi)=\max_{\ket{\psi}\in\bH_0}\bra{\psi}\cK_{\phi}\ket{\psi},
\qquad
\sigma_{\min}(\cK_\phi)=\min_{\ket{\psi}\in\bH_0}\bra{\psi}\cK_{\phi}\ket{\psi}.
\]

Note that the quantity  $\bra{\psi}\cK_{\phi}\ket{\psi}= {\rm Tr} \big(\cK_{\phi}\ket{\psi} \bra{\psi}\big)$ can be interpreted physically as the {\it energy} of a quantum system at state $\ket{\psi}$ when the system Hamiltonian is given by $\cK_{\phi}$. We refer to the quantity as the \emph{privacy energy}. According to Theorem \ref{thm:ldp-characterization}, the optimal privacy leakage level is determined by the maximum and minimum admissible privacy energies within the entanglement-constrained set $\bH_s$.

%The key distinction in our setting is that the optimization domain is restricted to the entanglement-constrained set $\bH_s$ rather than all quantum states. As a consequence, the privacy budgets $J_{\max}(\cK_{\phi},s)$ and $J_{\min}(\cK_{\phi},s)$ are obtained by optimizing the same linear functional over a strictly smaller feasible set. The restriction leads to two qualitatively different regimes for a centain $s$,
%\begin{itemize}
 %   \item If the extremal states attaining $J_{\max}$ and $J_{\min}$ in the unrestricted problem
  %  already lie in $\bH_s$, then we have
   % \begin{align*}
    %    J_{\max}(\cK_{\phi},s) = \sigma_{\max} (\cK_{\phi}), \quad J_{\min}%(\cK_{\phi},s) = \sigma_{\min}(\cK_{\phi})
 %   \end{align*}
 %   In this case, every entanglement level $s$ does not affect the QLDP level.
  %  \item If the entropy constraint excludes these extremal states, the optimal values will become
   % \begin{align*}
    %   J_{\max}(\cK_{\phi},s) < \sigma_{\max} (\cK_{\phi}) \quad \text{or}\quad J_{\min}(\cK_{\psi},s) > \sigma_{\min}(\cK_{\phi})
    %\end{align*}
    % In this regime, the entanglement level $s$ fundamentally enhances the QLDP level against local measurements.
% \end{itemize}

\subsection{The Riemannian Gradient of the Privacy Energy}

While the energy $\bra{\psi}\cK_{\phi}\ket{\psi}$ depends linearly on both the state $\ket{\psi}\bra{\psi}$ and the operator $\cK_{\phi}$, the admissible set $\bH_s$ of entangled states is non-convex. As a result, the maximally and minimally admissible privacy energies are constrained by the geometry of $\bH_s$. To capture this geometry explicitly, we introduce a parametrization of the energy function based on the Schmidt decomposition and unitary group actions.\footnote{For a Hilbert space $\cH$ with $\dim\cH=M$, the associated unitary group is defined as $\mathrm{U}(\cH):=\{U\in\mathbb C^{M\times M}\mid U^{\dg}U=UU^{\dg}=I\}$.}

For any bipartite pure state $\ket{\psi}\in\bH_s$, there exists a Schmidt decomposition
of the form $\ket{\psi} = \sum_{j=1}^N \sqrt{\lambda_j}\ket{j_A}\otimes\ket{j_B}$,
where $\{\ket{j_A}\}_{j=1}^N$ and $\{\ket{j_B}\}_{j=1}^N$ are orthonormal bases of $\cH_A$ and $\cH_B$, respectively. The vector $\lmd:=(\lmd_1,\lmd_2,\ldots,\lmd_N)$ collects the Schmidt coefficients of $\ket{\psi}$ and satisfies $\lmd_j\ge 0$ and $\sum_{j=1}^N \lmd_j=1$. If $\ket{\psi}$ has Schmidt rank strictly smaller than $N$, we implicitly pad $\lmd$ with zeros so that it can always be regarded as an element of $\mathbb R^N$. Moreover, for any $\ket{\psi}\in\bH_s$, its Schmidt coefficients $\lambda$ satisfies  $  {\bfE}(\psi )=  \bfH(\lmd)$, 
where $\bfH(\lmd):=-\sum_{j=1}^N \lmd_j\log\lmd_j$ denotes the Shannon entropy. We may also express the local operators $\cK_{\phi,a}$ and $\cK_{\phi,b}$ in the Schmidt bases $\{\ket{j_A}\}_{j=1}^N$ and $\{\ket{j_B}\}_{j=1}^N$, respectively. Equivalently, this corresponds to applying unitary rotations to the operators, leading to the representations
\begin{align}
\widetilde{\cK}_A := U_A\cK_{\phi,a}U_A^{\dg},
\qquad
\widetilde{\cK}_B := U_B\cK_{\phi,b}U_B^{\dg},
\end{align}
with $U_A\in\UHA$ and $U_B\in\UHB$. Under this parametrization, the energy $\bra{\psi}\cK_{\phi}\ket{\psi}$ can be written as
\begin{equation}
\bra{\psi}\cK_{\phi}\ket{\psi}
=\Tr\bigl(\Lambda^{1/2}\widetilde{\cK}_A\Lambda^{1/2}\widetilde{\cK}_B\bigr),
\qquad \Lambda:=\diag(\lmd_1,\ldots,\lmd_N).
\end{equation}
Consequently, computing $J_{\max}$ (or $J_{\min}$) reduces to maximizing (or minimizing)
a function on the product manifold
$\mathcal P:=\mathbb R^N\times\UHA\times\UHB$
\[
\bfF:\mathcal P\to\mathbb R,\qquad
(\lmd,U_A,U_B)\mapsto
\Tr\bigl(\Lambda^{1/2}\widetilde{\cK}_A\Lambda^{1/2}\widetilde{\cK}_B\bigr),
\]
subject to the constraints $\lmd_j\ge 0$, $
\sum_{j=1}^N \lmd_j=1$, and $
\bfH(\lmd)\ge s$. Then any local optimum of the optimization problem on the product manifold $\mathcal P$ admits a characterization via the Karush--Kuhn--Tucker (KKT) conditions within the framework of Riemannian optimization~\citep{yang2014optimality}. In particular, the stationarity with respect to the unitary variables $U_A$ and $U_B$ is characterized by the vanishing of the {\it Riemannian gradient} on the unitary group, where the Riemannian gradient of a function $f$ defined on the unitary group $\mathrm{U}(\cH)$ at $U\in\mathrm{U}(\cH)$ is  the orthogonal projection of the Euclidean gradient onto the tangent space, namely,
\[
\grad_U f = U\skew\left(U^\dg \nabla f(U)\right),
\]
where $\nabla f(U)$ denotes the Euclidean gradient of $f$ viewed as a function on the ambient space $\mathbb C^{n\times n}$ and $\skew(X):=\tfrac12(X-X^\dg)$.

\begin{theorem}[Riemannian Optimality Conditions for Privacy Energies]\label{thm:manifold-opt-Jmax}
Under the parametrization over $\mathbb R^N \times \UHA \times \UHB$, any local maximizer $(\lmd^*,\UAs,U_B^*)$ of  $J_{\max}(\cK_{\phi},s)$ satisfies the following KKT conditions with Lagrange multipliers $\nu^*,\xi^*\in \mathbb R, \eta^* \in \mathbb R^N$.
\begin{widetext}
\begin{align}
\textsc{(ST-U)}\qquad &\grad_{U_A} \bfF(\lmd^*,\UAs,U_B^*)=0, \quad \grad_{U_B}\bfF(\lmd^*,\UAs,U_B^*)=0; \label{eq:ST-U}
\\
\textsc{(ST-$\lambda$)} \qquad &\bigl[\nabla_{\lmd} \bfF(\lmd^*,\UAs,U_B^*)\bigr]_j = \nu^*-\eta_j^*+\xi^*(1+\log\lmd_j^*), \quad j=1,\ldots,N. \label{eq:ST-lambda-max} \\
\textsc{(PF)} \qquad &\lmd_j^*\ge 0,\quad \sum_{j=1}^N \lmd_j^*=1,\quad \bfH(\lmd^*)\ge s, \quad j=1,\ldots,N.  \label{eq:PF} \\
\textsc{(DF)} \qquad &\eta_j^*\ge 0,\quad \xi^*\ge 0 \quad j=1,\ldots,N.\label{eq:DF} \\
\textsc{(CS)} \qquad & \eta_j^*\lmd_j^*=0,\quad \xi^*\bigl(\bfH(\lmd^*)-s\bigr)=0 \quad j=1,\ldots,N. \label{eq:CS}
\end{align}
Similarly, for any local minimizer $(\lmd^*,\UAs,U_B^*)$ of $J_{\min}(\cK_{\phi},s)$, the same conditions~\eqref{eq:ST-U}--\eqref{eq:CS} hold, except that the sign in \eqref{eq:ST-lambda-max} is reversed, namely
\begin{align}
\bigl[\nabla_{\lmd} \bfF(\lmd^*,\UAs,U_B^*)\bigr]_j=\nu^*+\eta_j^*-\xi^*(1+\log\lmd_j^*), \quad j=1,\ldots,N. \label{eq:ST-lambda-min}
\end{align}
\end{widetext}
\end{theorem}

\begin{proof}
    It follows Theorem 4.1 in~\cite{yang2014optimality}.
\end{proof}

\subsection{Bounding the Privacy Energy}
We next derive an exact expression for the maximal privacy energy $J_{\max}(\cK_\phi,s)$ and a computable lower bound $\widetilde J_{\min}(\cK_\phi,s)$ on the minimal privacy energy $J_{\min}(\cK_\phi,s)$. Together with Theorem~\ref{thm:ldp-characterization}, these results yield a certified upper bound on the optimal privacy budget $\varepsilon^\star(s)$ and reveal how this bound depends on the entanglement entropy $s$.

We first derive a closed-form expression for $J_{\max}(\cK_{\phi},s)$ by solving the KKT conditions~\eqref{eq:ST-U}--\eqref{eq:CS}. To this end, we introduce the following notation. Since $\cE_A$ is a completely positive trace-preserving (CPTP) map, its adjoint $\cE_A^\dg$ is completely positive and unital. Consequently, the induced operator $\cK_a:=\cE_A^\dg(\cM_a)$ is positive semidefinite.
Without loss of generality, we let $\cK_a =\diag(\alpha_1,\dots,\alpha_N),$ where $\alpha_j$ denote the $j$-th largest eigenvalue of $\cK_a$, i.e., $\alpha_j=\sigma_j(\cK_a)$ and $\alpha_1\ge \alpha_2\ge\cdots\ge \alpha_N\ge 0$. Similarly, for subsystem $B$, we let $\cK_b = \diag(\beta_1,\dots,\beta_N) $, where $\beta_j:=\sigma_j(\cK_b)$ and $\beta_1\ge \beta_2\ge\cdots\ge \beta_N\ge 0$. We further define $\mu_i:=\alpha_i\beta_i$ and denote the degeneracy of the maximal eigenvalue by
\[
    d_{\max} := \#\{j:\mu_j=\mu_1 = \sigma_{\max} (\cK_{\phi})\}.
\]
The following theorem gives a closed-form expression for $J_{\max}(\mathcal K_{\phi},s)$ depending on the entropy constraint $s$ and $d_{\max}$.

\begin{theorem}[Phase-transition of Maximal Privacy Energy]\label{thm:max-product}
The value $J_{\max}(\mathcal K_{\phi},s)$ takes the following closed-form expression: 
\begin{itemize}
\item[(i)] \emph{(Low-entanglement).}
If $s \le \log d_{\max}$, then we have $J_{\max}(\mathcal K_{\phi},s) = \mu_1 = \sigma_{\max}(\mathcal K_{\phi})$, where   the optimal
solution is given by
\begin{align*}
\UAs = U_B^*&=I_{\cH},\\
\lambda_j^*&=
\begin{cases}
    \dfrac{1}{d_{\max}}, & j=1,\dots,d_{\max},\\[0.5ex]
    0, & j>d_{\max}.
\end{cases}
\end{align*}
 
\item[(ii)] \emph{(High-entanglement).}
If $s > \log d_{\max}$, then $J_{\max}(\mathcal K_{\phi},s)
=
\frac{1}{Z(\gamma)}
\sum_{j=1}^N \mu_j e^{\gamma \mu_j}
< \sigma_{\max}(\mathcal K_{\phi})$, where the optimal solution is given by the Gibbs form
\begin{align*}
\UAs=U_B^*=I_{\cH}, \quad
\lambda_j^* =\frac{\exp(\gamma\mu_j)}{Z(\gamma)}.
\end{align*}
Here, $\gamma>0$ is uniquely determined by the following entropy equation
\begin{equation}\label{eq:entropy-eq-max}
\log Z(\gamma)-\frac{\gamma}{Z(\gamma)}\sum_{j=1}^N\mu_j e^{\gamma\mu_j}=s,
\qquad Z(\gamma):=\sum_{j=1}^N e^{\gamma\mu_j}.
\end{equation}

\item[(iii)] \emph{(Strict monotonicity).}
For any $\log d_{\max}<s_1<s_2\le \log N$, we have
\[
J_{\max}(\mathcal K_{\phi},s_2)
<J_{\max}(\mathcal K_{\phi},s_1).
\]

\end{itemize}
\end{theorem}

The proof of Theorem~\ref{thm:max-product} is given in Appendix~\ref{app:thm:max-product}. This result reveals that  $J_{\max}(\cK_{\phi},s)$ exhibits a phase transition around the critical point $\log d_{\max}$. Specifically, when $s\le \log d_{\max}$, the entropy constraint is inactive and the optimal Schmidt coefficient concentrates uniformly on the $d_{\max}$-dimensional dominant eigenspace. In this regime, the maximal privacy energy coincides with the unconstrained optimum over the full space $\bH_0$, and hence remains constant at the spectral maximum $J_{\max}(\cK_{\phi},s)=\sigma_{\max}(\cK_{\phi})$.

In contrast, once the entropy level exceeds the threshold $s>\log d_{\max}$, the entropy constraint becomes active and forces the Schmidt weights to spread across the spectrum. The optimizer has a Gibbs-type exponential form arising from the entropy-constrained variational problem~\citep{jaynes1957information}. As a consequence, the value of $J_{\max}(\cK_{\phi},s)$ is strictly smaller than $\sigma_{\max}(\cK_{\phi})$ and becomes a smooth, monotonically decreasing function of $s$. Moreover, in the maximally entangled case, i.e., $s=\log N$, the Gibbs weights become uniform on all indices of $\lambda$ and $J_{\max}(\cK_{\phi},s)$ attains its minimal value $\bar{\mu}=\frac{1}{N}\sum_{j=1}^N \mu_j$, corresponding to the complete delocalization of spectral alignment.

Next, the following theorem provides an explicitly computable lower bound for $J_{\min}(\mathcal K_{\phi},s)$, namely $\widetilde{J}_{\min}(\mathcal K_{\phi},s)$. We denote $\widetilde{\mu}_i := \alpha_{i}\beta_{N}$, and the degeneracy of the minimal eigenvalue by
\[ d_{\min} := \#\{ j : \widetilde{\mu}_j=\widetilde{\mu}_N = \sigma_{\min}(\cK_{\phi}) \}.
\]
Then $\widetilde{J}_{\min}(\mathcal K_{\phi},s)$ presents a similar phase-transition with energy level $\{\widetilde{\mu}_j\}_{j=1}^N$.

\begin{theorem}[Phase-transition of Relaxed Minimal Privacy Energy]\label{thm:min-product}
The privacy energy $J_{\min}(\mathcal K_{\phi},s)$ admits a lower bound of $\widetilde{J}_{\min}(\mathcal K_{\phi},s)$ which exhibits the following phase transition:

\begin{itemize}
\item[(i)] \emph{(Low-entanglement).}
If $s \le \log d_{\min}$, then 
$J_{\min}(\mathcal K_{\phi},s) \ge \widetilde{J}_{\min}(\mathcal K_{\phi},s) = \widetilde\mu_N = \sigma_{\min}(\mathcal K_{\phi})$.
\item[(ii)] \emph{(High-entanglement).}
If $s > \log d_{\min}$, we have
\[
J_{\min}(\mathcal K_{\phi},s)
\ge \widetilde{J}_{\min}(\mathcal K_{\phi},s)
=\frac{1}{\widetilde Z(\widetilde\gamma)}
\sum_{j=1}^N\widetilde\mu_j e^{-\widetilde\gamma\widetilde\mu_j}
>\sigma_{\min}(\mathcal K_{\phi}),
\]
where $\widetilde\gamma>0$ is uniquely determined by the entropy equation
\begin{equation*}\label{eq:entropy-eq-clean}
\log \widetilde Z(\widetilde\gamma)
+\frac{\widetilde\gamma}{\widetilde Z(\widetilde\gamma)}
\sum_{j=1}^N\widetilde\mu_j e^{-\widetilde\gamma\widetilde\mu_j}=s,
\qquad
\widetilde Z(\widetilde\gamma):=\sum_{j=1}^N e^{-\widetilde\gamma\widetilde\mu_j}.
\end{equation*}
\item[(iii)] \emph{(Strict monotonicity).}
For any $\log d_{\min}<s_1<s_2\le \log N$, we have
\[
\widetilde J_{\min}(\mathcal K_{\phi},s_2)
>\widetilde J_{\min}(\mathcal K_{\phi},s_1).
\]
\end{itemize}
\end{theorem}

The proof of Theorem~\ref{thm:min-product} is given in Appendix~\ref{app:thm:min-product}. The fact that we obtained a closed-form solution to $J_{\max}(\mathcal K_{\phi},s)$, but only a relaxed lower bound for $J_{\min}(\mathcal K_{\phi},s)$, is rooted in the geometry of the underlying optimization problem. In both cases, we are optimizing a linear functional over the non-convex feasible set $\bH_s$. For the maximization problem, this non-convexity is benign and the maximum can be attained at a vertex of the set, which leads to the closed-form expression in terms of the products $\alpha_i\beta_i$ and the associated Gibbs weights. In contrast, for the minimization problem, the true minimum in general is not attained at such vertices. Instead, it typically lies in the interior of the feasible set and cannot be represented solely by a simple rearrangement of eigenvalues (such as pairing $\alpha_i$ with $\beta_{N-i+1}$). As a consequence, we construct a tractable convex relaxation that yields a computable lower bound $\widetilde J_{\min}(\mathcal K_{\phi},s)$, but this relaxation does not generally coincide with the exact minimum $J_{\min}(\mathcal K_{\phi},s)$.

%Theorem~\ref{thm:min-product} shows that the minimum value $J_{\min}(\mathcal K_{\phi},s)$ does not, in general, admit a closed-form expression. Nevertheless, it can be effectively controlled by a computable lower bound $\widetilde{J}{\min}(\mathcal K{\phi},s)$ obtained from a linear program with an entropy constraint.

%Similar to the behavior of $J_{\max}$, this lower bound $\widetilde{J}{\min}(\mathcal K{\phi},s)$ exhibits a threshold phenomenon governed by $\log d_{\min}$. In the low-entanglement regime ($s\le \log d_{\min}$), the entropy constraint is inactive and the lower bound coincides with the spectral minimum $\sigma_{\min}(\mathcal K_{\phi})$, indicating that limited entanglement does not improve the minimum value.

%Once the entropy level exceeds this threshold, the entropy constraint becomes active and yields a strictly larger lower bound $\widetilde{J}_{\min}(\mathcal K_{\phi},s)>\sigma_{\min}(\mathcal K_{\phi})$. Moreover, as $s$ increases, $\widetilde{J}_{\min}(\mathcal K_{\phi},s)$ increases monotonically. At $s=\log N$, the relaxed value is $\frac{\beta_N}{N}\sum_{i=1}^N\alpha_i$, while the exact minimum over maximally entangled states is $\frac{1}{N}\sum_{i=1}^N\alpha_i\beta_{N-i+1}$.

\subsection{The Geometry of Quantum Differential Privacy}
We now establish Theorem~\ref{thm-main} and show how the geometry of the entanglement-constrained space $\mathbb{H}_s$ leads to a phase transition in quantum differential privacy, combining the results in Theorems~\ref{thm:ldp-characterization}--\ref{thm:min-product}.
\begin{itemize}
\item[(i)] Low entanglement does not change the QLDP.

When the entanglement entropy satisfies

\[
s \le \min\{\log d_{\max}, \log d_{\min}\},
\]

the entropy constraint does not exclude the extremal eigenstates of $\cK_\phi$. In this case, the maximal and minimal privacy energies are still attained at the largest and smallest eigenvalues of $\cK_\phi$, yielding

\[
\varepsilon^*(s)
=\log\max_{\substack{\ket{\phi_a}\in \cH_A\\
\ket{\phi_b}\in \cH_B}}
\frac{\sigma_{\max}(\cK_\phi)}{\sigma_{\min}(\cK_\phi)}
= \varepsilon^*(0).
\]

Thus, the privacy leakage level coincides with the entanglement-free case, recovering the QLDP characterization in~\cite{guan2025optimal}.

\item[(ii)] High entanglement enhances the QLDP.

When the entanglement entropy exceeds the threshold,
\[
s>s_0,
\]
the entropy constraint becomes active and excludes at least one of the spectral configurations responsible for the unentangled worst-case leakage. Theorem~\ref{thm:max-product} gives an exact Gibbs-type characterization of $J_{\max}(\cK_\phi,s)$. For the minimal privacy energy, Theorem~\ref{thm:min-product} provides the computable lower bound
\[
J_{\min}(\cK_\phi,s)
\ge \widetilde J_{\min}(\cK_\phi,s).
\]
We therefore define
\[
\overline{\varepsilon}(s)
:=
\log\max_{\substack{\ket{\phi_a}\in\cH_A\\
                      \ket{\phi_b}\in\cH_B}}
\frac{J_{\max}(\cK_\phi,s)}
     {\widetilde J_{\min}(\cK_\phi,s)}.
\]
Using Theorem~\ref{thm:ldp-characterization} and the lower bound on $J_{\min}$, we obtain
\[
\begin{aligned}
\varepsilon^\star(s)
&=
\log\max_{\substack{\ket{\phi_a}\in\cH_A\\
                      \ket{\phi_b}\in\cH_B}}
\frac{J_{\max}(\cK_\phi,s)}
     {J_{\min}(\cK_\phi,s)}\\
&\le
\log\max_{\substack{\ket{\phi_a}\in\cH_A\\
                      \ket{\phi_b}\in\cH_B}}
\frac{J_{\max}(\cK_\phi,s)}
     {\widetilde J_{\min}(\cK_\phi,s)}\\
&=\overline{\varepsilon}(s) < \varepsilon^\star(0) .
\end{aligned}
\]
Moreover, in the active high-entanglement regime, the energy $J_{\max}(\cK_\phi,s)$ decreases while the relaxed energy $\widetilde J_{\min}(\cK_\phi,s)$ increases. Consequently, the certified upper bound $\overline{\varepsilon}(s)$ is strictly decreasing for any $s>s_0$. Thus, increasing entanglement enhances privacy, as certified by an explicitly computable upper bound that decreases strictly with the entanglement entropy.

In particular, consider the case where $\cK_a=\cE_A^{\dg}(\ket{\phi_a}\bra{\phi_a})$ has a vanishing smallest eigenvalue, $\alpha_N=0$, while $\cK_b=\cE_B^{\dg}(\ket{\phi_b}\bra{\phi_b})$ satisfies $\beta_N>0$. In the unentangled regime, $\sigma_{\min}(\cK_\phi)=0$, yielding infinite privacy leakage. Under a sufficiently strong entanglement constraint, the Gibbs weights defining $\widetilde J_{\min}(\cK_\phi,s)$ are strictly positive. Provided that not all coefficients $\widetilde\mu_j=\alpha_j\beta_N$ vanish, we have
\[
J_{\min}(\cK_\phi,s)
\ge\widetilde J_{\min}(\cK_\phi,s)
>0.
\]
More generally, whenever this positivity holds uniformly over the local measurements, $\overline{\varepsilon}(s)<+\infty$, and therefore
\[
\varepsilon^\star(s)
\le\overline{\varepsilon}(s)
<+\infty.
\]
Consequently, entanglement can render an otherwise non-private mechanism private over the entanglement-constrained domain.
\end{itemize}

The above discussion highlights the central role of entanglement, or non-classical quantum correlation, in the differential privacy of quantum computing methods. Weak entanglement does not affect the achievable privacy level, whereas sufficiently strong entanglement yields a strict improvement over the entanglement-free privacy level and provides a progressively decreasing certified upper bound. Entanglement can even turn some non-private mechanisms into private ones. This result stands in sharp contrast to classically correlated computing settings, where correlations typically degrade privacy guarantees. Our results show that, rather than being detrimental, quantum entanglement can serve as an intrinsic resource for privacy protection.

In addition, the certified dependence of privacy on entanglement exhibits nonlinear behavior. In contrast to prior QLDP results~\citep{guan2025optimal}, which are governed by linear spectral quantities such as maximal and minimal eigenvalues, the nonlinearity here arises from the nonconvex geometry of the entanglement-constrained set $\bH_s$. We parametrize this set using the Schmidt decomposition and analyze the resulting optimization problem through Riemannian optimization. This geometric perspective potentially paves the way for a deeper understanding of the role of entanglement in quantum information, quantum computing, and privacy protection.

\section{Numerical Experiments}\label{sec:exp}

To illustrate our theoretical results, we conduct two numerical experiments in this section. The first experiment demonstrates the phase-transition phenomenon as characterized in Theorem~\ref{thm-main}, and the second one shows that entanglement can transform a mechanism that is not private into one that satisfies differential privacy. Throughout the section, we set $\mathcal H_{AB}$ as a $4$-qubit bipartite quantum system, where each subsystem consists of two qubits, i.e., $\mathcal H_A\cong(\mathbb C^2)^{\otimes 2}$ and $\mathcal H_B\cong(\mathbb C^2)^{\otimes 2}$.

\subsection{Two-qubit Block-depolarizing Channel}
We consider the following mechanism on both $\cH_A$ and $\cH_B$ 
\begin{equation}\label{eq:Ebeta}
\mathcal N_\beta(\rho)
=(1-\beta)\Tr(P_0\rho)\frac{P_0}{2}
+(1-\beta)\Tr(P_1\rho)\frac{P_1}{2}
+\beta\frac{I_4}{4},\qquad \beta\in[0,1].
\end{equation}
where $P_0 := \ket{00}\!\bra{00} + \ket{11}\!\bra{11}$ is the orthogonal projector onto the even subspace of $(\mathbb C^2)^{\otimes 2}$ and $P_1 := I_4 - P_0$ is the projector onto the odd subspace. By construction, $\mathcal N_\beta$ is a CPTP map.

For a rank-one  POVM element $\cM=\ket{\psi}\!\bra{\psi}$, the adjoint of $\mathcal N_\beta$ admits the explicit form
\[
\mathcal N_\beta^\dagger(\cM)
=(1-\beta)\Bigl(\frac{t}{2}P_0+\frac{1-t}{2}P_1\Bigr)
+\beta\frac{\Tr(\cM)}{4}I_4.
\]
where the weight $t = \Tr(P_0 \cM)=\langle\psi|P_0|\psi\rangle \in [0,1]$. Since both $P_0$ and $P_1$ have rank $2$, the spectrum of $\mathcal N_\beta^\dagger(\cM)$ consists of two identical maximal eigenvalues and two minimal eigenvalues
\begin{align}\label{eq:spec-depor}
\mathrm{spec}\bigl(\mathcal N_\beta^\dagger(\cM)\bigr)
= \{ v(t), v(t), w(t), w(t)\},
\end{align}
where $v(t) = (1-\beta)\frac{t}{2}+\frac{\beta}{4}$  and $w(t) = 1/2-v(t)$.
\subsection{Phase Transition of the QLDP Leakage Level}
To show the phase-transition in Theorem~\ref{thm-main}, we consider a symmetric local quantum mechanism $\cE_A=\cE_B= \mathcal N_\beta$ with $\beta=\tfrac12$, and evaluate the optimal ECLM-QLDP leakage level $\varepsilon^*$ as a function of the entanglement entropy $s$. 

From~\eqref{eq:spec-depor}, we know that the mechanism $\cE =\mathcal N_\beta \otimes \mathcal N_\beta $ consists of two degenerate maximal eigenvalues and two degenerate minimal eigenvalues for every POVM element $\cM=\ket{\phi}\!\bra{\phi}$. Consequently, the phase-transition threshold is $s=\log 2$. Considering separately the regimes $s\le\log 2$ and $s>\log 2$, we obtain the following certified upper bound on the optimal ECLM-QLDP leakage level.
\begin{proposition}\label{prop:qldp-exp-1}
According to Theorems~\ref{thm:max-product} and~\ref{thm:min-product}, the optimal ECLM-QLDP leakage level $\varepsilon^\star(s)$ of $\cE$ satisfies
\begin{equation}
\begin{aligned}
\varepsilon^\star(s)&\le \overline{\varepsilon}(s),\\
\overline{\varepsilon}(s)&:=
\begin{cases}
2\log 3,
& s\le \log 2,\\[0.5ex]
\displaystyle\log\frac{8\tau+1}{3-2\tau},
& s>\log 2.
\end{cases}
\end{aligned}
\end{equation}
where $\tau\in[1/2,1]$ is uniquely determined by
\[
s=\log 2-\tau\log\tau-(1-\tau)\log(1-\tau),
\]
and decreases monotonically as the entanglement entropy increases.
\end{proposition}

\begin{figure}[t]
        \centering
        \subfloat[Phase transition of the certified QLDP leakage upper bound and the numerically computed optimal QLDP leakage level.\label{fig:QLDPvsEntan-1}]{
                \includegraphics[width=0.48\linewidth]{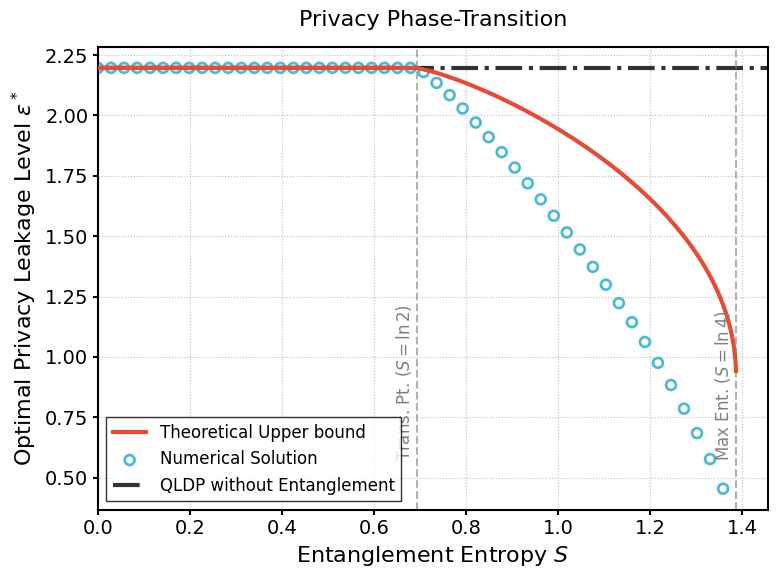}
        }
        \hfill
        \subfloat[Finite QLDP leakage achieved via entanglement for a non-private mechanism.\label{fig:QLDPvsEntan}]{
                \includegraphics[width=0.48\linewidth]{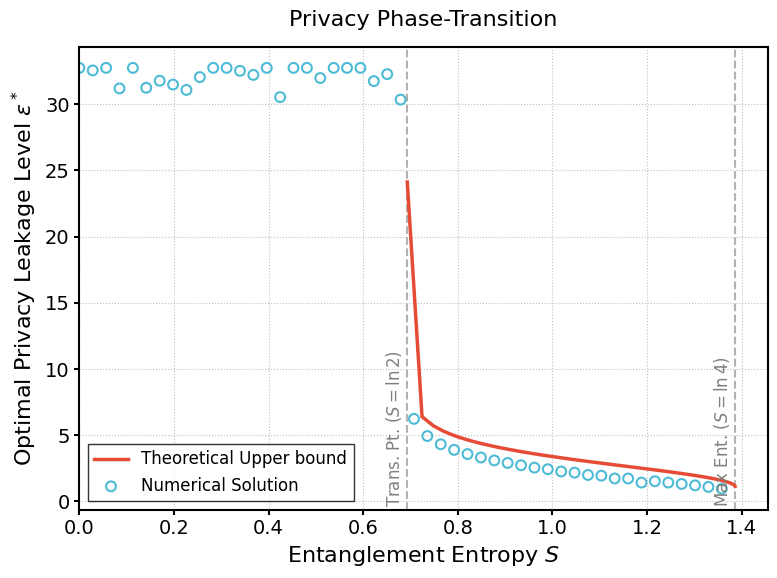}
        }
        \caption{Entanglement-dependent QLDP leakage.}
\end{figure}

We plot the certified QLDP leakage upper bound $\overline{\varepsilon}(s)$ derived in Proposition~\ref{prop:qldp-exp-1} as a function of the entanglement entropy $s$ in Figure~\ref{fig:QLDPvsEntan-1}. We also report the numerically computed optimal leakage level $\varepsilon^\star(s)$ obtained by solving the Riemannian optimization problem in Theorem~\ref{thm:manifold-opt-Jmax}. In this example, both curves exhibit the same qualitative phase-transition behavior. Below the critical threshold $s=\log 2$, the certified upper bound coincides with the unentangled privacy level, and the numerically computed optimal leakage level remains unchanged. Above the threshold, the certified upper bound strictly decreases with $s$, while the numerically computed optimal leakage level displays the same decreasing trend and attains its smallest observed value at $s=\log 4$. The numerical results are therefore consistent with the certified privacy improvement established by the theoretical upper bound.

\subsection{Entanglement-Induced Privatization of Non-Private Mechanisms}

In this subsection, we show that entanglement can convert a non-private local mechanism into a private one. We consider asymmetric local mechanisms given by $\cE'_A=\mathcal N_0$ and $\cE'_B=\mathcal N_{1/2}$. For the measurement operator $\cM_a=\ket{00}\!\bra{00}$, the adjoint action yields
\begin{align*}
\mathrm{spec}\!\bigl(\mathcal N_0^\dagger(\cM_a)\bigr)
=\left\{\tfrac12,\tfrac12,0,0\right\}.
\end{align*}
In particular, $\mathcal N_0^\dagger(\cM_a)$ admits a zero eigenvalue, implying that the local mechanism $\cE'_A$ is not QLDP in the entanglement-free domain~\citep{guan2025optimal}, where the corresponding leakage level is $\varepsilon_a^*(0)=+\infty$. The same conclusion holds for the product mechanism $\cE'=\cE'_A\otimes\cE'_B$, and therefore $\cE'$ is not QLDP in unentangled cases.

We now compute an upper bound on the ECLM-QLDP leakage level of $\cE'=\cE'_A\otimes\cE'_B$ as a function of the entanglement entropy $s$. By applying Theorems~\ref{thm:max-product} and~\ref{thm:min-product}, and following the same derivation as in Proposition~\ref{prop:qldp-exp-1}, we obtain the corresponding upper bound for the leakage level $\varepsilon^*(s)$ in the following proposition.

\begin{proposition}\label{prop:qldp-exp-2}
According to Theorems~\ref{thm:max-product} and~\ref{thm:min-product}, the optimal ECLM-QLDP leakage level $\varepsilon^*(s)$ of $\cE'$ has the following upper bound
\begin{align}
    \nonumber \varepsilon^\star(s)&\le \overline{\varepsilon}(s),\\
    \overline{\varepsilon}(s)&:=
    \begin{cases}
        \infty, &\quad s\le \log2 ;\\
        \log \frac{3\tau}{1-\tau}, & \quad s>\log2.
    \end{cases}
\end{align}
Here, $\tau \in [1/2,1]$ is uniquely determined by the equation $s=\log 2-\tau\log\tau-(1-\tau)\log(1-\tau)$ and decreases monotonically as the entanglement entropy increases.
\end{proposition}

The results show that, below the phase-transition threshold, the mechanism remains non-private, matching its behavior in the unentangled setting. Once the entanglement entropy exceeds the threshold, Proposition~\ref{prop:qldp-exp-2} provides a finite upper bound on $\varepsilon^\star(s)$, which certifies that the mechanism has finite QLDP leakage. Figure~\ref{fig:QLDPvsEntan} compares this certified upper bound with the numerically computed optimal leakage level. The numerical results are consistent with the theoretical guarantee and illustrate how a sufficiently strong entanglement constraint can transform a non-private mechanism into a private one.

\section{Conclusion}\label{sec:con}
In this work, we established a quantitative relationship between quantum entanglement and quantum local differential privacy. We first considered a quantum processing model with entangled inputs, where both the mechanism and measurements were restricted to local operations to avoid additional entanglement. Within this framework, we then defined entanglement-constrained $\varepsilon$-QLDP against local measurement adversaries (EC-LM-$\varepsilon$-QLDP). By capturing the nonconvex geometry of the entanglement-constrained set $\bH_s$ via manifold parametrization and using the KKT conditions in Riemannian optimization, we characterized the dependence of EC-LM-$\varepsilon$-QLDP on the entanglement entropy, and obtained an explicit bound for the QLDP leakage level. Finally, several numerical experiments illustrated our theoretical results. Our central finding is a phase-transition phenomenon of QLDP with entanglement: insufficient entanglement leaves the QLDP leakage level unchanged, whereas sufficiently strong entanglement makes the exact leakage strictly smaller than its unentangled value and yields an explicitly computable upper bound that decreases strictly with the entanglement entropy. This bound provides progressively stronger certified privacy guarantees and can certify that some non-private mechanisms become private under sufficiently strong entanglement constraints.

Future work will investigate the connection between quantum non-locality, the Heisenberg uncertainty principle, and quantum differential privacy. Another future direction we aim to develop is robust privacy-preserving protocols that exploit entanglement to enhance privacy guarantees in quantum communication and quantum machine learning.

\section*{Author contribution statements}
G. Shi and P. Sadeghi conceived the original research problem. X. Wang developed the theoretical analysis and derived the main results. G. Shi and P. Sadeghi verified the theoretical derivations and validated the computational results. G. Shi supervised the project and reviewed the manuscript. 

A large language model (LLM) was used solely to assist with grammar correction and language editing of the manuscript. The authors reviewed and edited all AI-generated suggestions and take full responsibility for the content of the manuscript.

\bibliographystyle{quantum}
\bibliography{ref}

@preamble{"% MakeDoiLinksExplicit"}

@book{bhatia2013matrix,
  title={Matrix analysis},
  author={Bhatia, Rajendra},
  volume={169},
  year={2013},
  publisher={Springer Science \& Business Media},
  doi={10.1007/978-1-4612-0653-8}
}

@InProceedings{dwork2006differential,
author="Dwork, Cynthia",
editor="Bugliesi, Michele
and Preneel, Bart
and Sassone, Vladimiro
and Wegener, Ingo",
title="Differential Privacy",
booktitle="Automata, Languages and Programming",
year="2006",
publisher="Springer Berlin Heidelberg",
address="Berlin, Heidelberg",
pages="1--12",
doi="10.1007/11787006_1",
}

@InProceedings{dwork2006calibrating,
author="Dwork, Cynthia
and McSherry, Frank
and Nissim, Kobbi
and Smith, Adam",
editor="Halevi, Shai
and Rabin, Tal",
title="Calibrating Noise to Sensitivity in Private Data Analysis",
booktitle="Theory of Cryptography",
year="2006",
publisher="Springer Berlin Heidelberg",
address="Berlin, Heidelberg",
pages="265--284",
doi="10.1007/11681878_14"
}

@article{Warner01031965,
author = {Stanley L. Warner},
title = {Randomized Response: A Survey Technique for Eliminating Evasive Answer Bias},
journal = {Journal of the American Statistical Association},
volume = {60},
number = {309},
pages = {63--69},
year = {1965},
publisher = {Taylor \& Francis},
doi = {10.1080/01621459.1965.10480775},
}

@inproceedings{evfimievski2003limiting,
author = {Evfimievski, Alexandre and Gehrke, Johannes and Srikant, Ramakrishnan},
title = {Limiting privacy breaches in privacy preserving data mining},
year = {2003},
isbn = {1581136706},
publisher = {Association for Computing Machinery},
address = {New York, NY, USA},
url = {https://doi.org/10.1145/773153.773174},
doi = {10.1145/773153.773174},
booktitle = {Proceedings of the Twenty-Second ACM SIGMOD-SIGACT-SIGART Symposium on Principles of Database Systems},
pages = {211–222},
numpages = {12},
location = {San Diego, California},
series = {PODS '03}
}

@article{Kifer2014pufferfish,
author = {Kifer, Daniel and Machanavajjhala, Ashwin},
title = {Pufferfish: A framework for mathematical privacy definitions},
year = {2014},
issue_date = {January 2014},
publisher = {Association for Computing Machinery},
address = {New York, NY, USA},
volume = {39},
number = {1},
issn = {0362-5915},
url = {https://doi.org/10.1145/2514689},
doi = {10.1145/2514689},
journal = {ACM Trans. Database Syst.},
month = jan,
articleno = {3},
numpages = {36}
}

@INPROCEEDINGS{zhou2017differential,
  author={Zhou, Li and Ying, Mingsheng},
  booktitle={2017 IEEE 30th Computer Security Foundations Symposium (CSF)}, 
  title={Differential Privacy in Quantum Computation}, 
  year={2017},
  volume={},
  number={},
  pages={249-262},
  doi={10.1109/CSF.2017.23}}

@inproceedings{aasornson2019gentle,
author = {Aaronson, Scott and Rothblum, Guy N.},
title = {Gentle measurement of quantum states and differential privacy},
year = {2019},
isbn = {9781450367059},
publisher = {Association for Computing Machinery},
address = {New York, NY, USA},
url = {https://doi.org/10.1145/3313276.3316378},
doi = {10.1145/3313276.3316378},
booktitle = {Proceedings of the 51st Annual ACM SIGACT Symposium on Theory of Computing},
pages = {322–333},
numpages = {12},
location = {Phoenix, AZ, USA},
series = {STOC 2019}
}

@ARTICLE{hirche2023quantum,
  author={Hirche, Christoph and Rouzé, Cambyse and França, Daniel Stilck},
  journal={IEEE Transactions on Information Theory}, 
  title={Quantum Differential Privacy: An Information Theory Perspective}, 
  year={2023},
  volume={69},
  number={9},
  pages={5771-5787},
  doi={10.1109/TIT.2023.3272904}}

@inproceedings{guan2025optimal,
author = {Guan, Ji},
title = {Optimal Mechanisms for Quantum Local Differential Privacy},
year = {2025},
isbn = {9798400715259},
publisher = {Association for Computing Machinery},
address = {New York, NY, USA},
url = {https://doi.org/10.1145/3719027.3765178},
doi = {10.1145/3719027.3765178},
booktitle = {Proceedings of the 2025 ACM SIGSAC Conference on Computer and Communications Security},
pages = {3737–3749},
numpages = {13},
location = {Taipei, Taiwan},
series = {CCS '25}
}

@article{watkins2023quantum,
  title={Quantum machine learning with differential privacy},
  author={Watkins, William M and Chen, Samuel Yen-Chi and Yoo, Shinjae},
  journal={Scientific Reports},
  volume={13},
  number={1},
  pages={2453},
  year={2023},
  publisher={Nature Publishing Group UK London},
  doi={10.1038/s41598-022-24082-z}
}

@ARTICLE{nuradha2024pufferfish,
  author={Nuradha, Theshani and Goldfeld, Ziv and Wilde, Mark M.},
  journal={IEEE Transactions on Information Theory}, 
  title={Quantum Pufferfish Privacy: A Flexible Privacy Framework for Quantum Systems}, 
  year={2024},
  volume={70},
  number={8},
  pages={5731-5762},
  doi={10.1109/TIT.2024.3404927}}

@book{nielsen2010quantum,
  title={Quantum computation and quantum information},
  author={Nielsen, Michael A and Chuang, Isaac L},
  year={2010},
  publisher={Cambridge university press},
  doi={10.1017/CBO9780511976667}
}

@article{clauser1969proposed,
  title={Proposed experiment to test local hidden-variable theories},
  author={Clauser, John F and Horne, Michael A and Shimony, Abner and Holt, Richard A},
  journal={Physical review letters},
  volume={23},
  number={15},
  pages={880},
  year={1969},
  publisher={APS},
  doi={10.1103/PhysRevLett.23.880}
}

@article{harrow2009quantum,
  title = {Quantum Algorithm for Linear Systems of Equations},
  author = {Harrow, Aram W. and Hassidim, Avinatan and Lloyd, Seth},
  journal = {Phys. Rev. Lett.},
  volume = {103},
  issue = {15},
  pages = {150502},
  numpages = {4},
  year = {2009},
  month = {Oct},
  publisher = {American Physical Society},
  doi = {10.1103/PhysRevLett.103.150502},
  url = {https://link.aps.org/doi/10.1103/PhysRevLett.103.150502}
}

@article{luis1996optimum,
  title = {Optimum phase-shift estimation and the quantum description of the phase difference},
  author = {Luis, A. and Pe\ifmmode \check{r}\else \v{r}\fi{}ina, J.},
  journal = {Phys. Rev. A},
  volume = {54},
  issue = {5},
  pages = {4564--4570},
  numpages = {0},
  year = {1996},
  month = {Nov},
  publisher = {American Physical Society},
  doi = {10.1103/PhysRevA.54.4564},
  url = {https://link.aps.org/doi/10.1103/PhysRevA.54.4564}
}

@inproceedings{yang2015bayesian,
  title={Bayesian differential privacy on correlated data},
  author={Yang, Bin and Sato, Issei and Nakagawa, Hiroshi},
  booktitle={Proceedings of the 2015 ACM SIGMOD international conference on Management of Data},
  pages={747--762},
  year={2015},
  doi={10.1145/2723372.2747643}
}

@article{zhang2022correlated,
  title={Correlated data in differential privacy: definition and analysis},
  author={Zhang, Tao and Zhu, Tianqing and Liu, Renping and Zhou, Wanlei},
  journal={Concurrency and Computation: Practice and Experience},
  volume={34},
  number={16},
  pages={e6015},
  year={2022},
  publisher={Wiley Online Library},
  doi={10.1002/cpe.6015}
}

@article{yang2014optimality,
  title={Optimality conditions for the nonlinear programming problems on {Riemannian} manifolds},
  author={Yang, Wei Hong and Zhang, Lei-Hong and Song, Ruyi},
  journal={Pacific Journal of Optimization},
  volume={10},
  number={2},
  pages={415--434},
  year={2014},
  url={https://www.researchgate.net/publication/259288927_Optimality_conditions_for_the_nonlinear_programming_problems_on_Riemannian_manifolds}
}

@article{kasiviswanathan2011can,
  title={What can we learn privately?},
  author={Kasiviswanathan, Shiva Prasad and Lee, Homin K and Nissim, Kobbi and Raskhodnikova, Sofya and Smith, Adam},
  journal={SIAM Journal on Computing},
  volume={40},
  number={3},
  pages={793--826},
  year={2011},
  publisher={SIAM},
  doi={10.1137/090756090}
}

@article{kairouz2014extremal,
  title={Extremal mechanisms for local differential privacy},
  author={Kairouz, Peter and Oh, Sewoong and Viswanath, Pramod},
  journal={Advances in neural information processing systems},
  volume={27},
  year={2014},
  eprint={1407.1338},
  archivePrefix={arXiv},
  primaryClass={cs.IT}
}

@inproceedings{erlingsson2014rappor,
  title={Rappor: Randomized aggregatable privacy-preserving ordinal response},
  author={Erlingsson, {\'U}lfar and Pihur, Vasyl and Korolova, Aleksandra},
  booktitle={Proceedings of the 2014 ACM SIGSAC conference on computer and communications security},
  pages={1054--1067},
  year={2014},
  doi={10.1145/2660267.2660348}
}

@article{li2021quantum,
  title={Quantum federated learning through blind quantum computing},
  author={Li, Weikang and Lu, Sirui and Deng, Dong-Ling},
  journal={Science China Physics, Mechanics \& Astronomy},
  volume={64},
  number={10},
  pages={100312},
  year={2021},
  publisher={Springer},
  doi={10.1007/s11433-021-1753-3}
}

@article{quek2021private,
  title={Private learning implies quantum stability},
  author={Quek, Yihui and Arunachalam, Srinivasan and Smolin, John A},
  journal={Advances in Neural Information Processing Systems},
  volume={34},
  pages={20503--20515},
  year={2021},
  eprint={2102.07171},
  archivePrefix={arXiv},
  primaryClass={cs.LG}
}

@article{du2021quantum,
  title={Quantum noise protects quantum classifiers against adversaries},
  author={Du, Yuxuan and Hsieh, Min-Hsiu and Liu, Tongliang and Tao, Dacheng and Liu, Nana},
  journal={Physical Review Research},
  volume={3},
  number={2},
  pages={023153},
  year={2021},
  publisher={APS},
  doi={10.1103/PhysRevResearch.3.023153}
}

@inproceedings{abadi2016deep,
  title={Deep learning with differential privacy},
  author={Abadi, Martin and Chu, Andy and Goodfellow, Ian and McMahan, H Brendan and Mironov, Ilya and Talwar, Kunal and Zhang, Li},
  booktitle={Proceedings of the 2016 ACM SIGSAC conference on computer and communications security},
  pages={308--318},
  year={2016},
  doi={10.1145/2976749.2978318}
}

@inproceedings{abowd2018us,
  title={The {U.S. Census Bureau} adopts differential privacy},
  author={Abowd, John M},
  booktitle={Proceedings of the 24th ACM SIGKDD International Conference on Knowledge Discovery \& Data Mining},
  pages={2867--2867},
  year={2018},
  doi={10.1145/3219819.3226070}
}

@inproceedings{yilmaz2022genomic,
  title={Genomic data sharing under dependent local differential privacy},
  author={Yilmaz, Emre and Ji, Tianxi and Ayday, Erman and Li, Pan},
  booktitle={Proceedings of the twelfth ACM conference on data and application security and privacy},
  pages={77--88},
  year={2022},
  doi={10.1145/3508398.3511519}
}

@inproceedings{liu2016dependence,
  title={Dependence makes you vulnerable: Differential privacy under dependent tuples.},
  author={Liu, Changchang and Chakraborty, Supriyo and Mittal, Prateek},
  booktitle={NDSS},
  volume={16},
  pages={21--24},
  year={2016},
  doi={10.14722/ndss.2016.23279}
}

@article{uhlmann1976transition,
  title={The “transition probability” in the state space of a*-algebra},
  author={Uhlmann, Armin},
  journal={Reports on Mathematical Physics},
  volume={9},
  number={2},
  pages={273--279},
  year={1976},
  publisher={Elsevier},
  doi={10.1016/0034-4877(76)90060-4}
}

@article{bell1964einstein,
  title = {On the {Einstein-Podolsky-Rosen} paradox},
  author = {Bell, J. S.},
  journal = {Physics Physique Fizika},
  volume = {1},
  issue = {3},
  pages = {195--200},
  numpages = {6},
  year = {1964},
  month = {Nov},
  publisher = {American Physical Society},
  doi = {10.1103/PhysicsPhysiqueFizika.1.195},
  url = {https://link.aps.org/doi/10.1103/PhysicsPhysiqueFizika.1.195}
}

@article{jaynes1957information,
  author = {Jaynes, E. T.},
  title = {Information Theory and Statistical Mechanics},
  journal = {Physical Review},
  volume = {106},
  number = {4},
  pages = {620--630},
  year = {1957},
  doi = {10.1103/PhysRev.106.620}
}

\onecolumn
\appendix
\newpage

\section{Notation Table}

For convenience, we summarize the key notations used throughout this paper in Table \ref{tab:notation}.
\begin{table}[h!]
\centering
\renewcommand{\arraystretch}{1.3} % 调整行高，使表格更美观
\begin{tabular}{p{0.2\textwidth} p{0.75\textwidth}}
\hline
\textbf{Symbol} & \textbf{Description} \\
\hline
\multicolumn{2}{l}{\textit{Spaces and Basic Operators}} \\
$\mathcal{H}$ & Hilbert space associated with a quantum system \\
$\mathcal{H}_A, \mathcal{H}_B$ & Local Hilbert spaces for subsystems A and B \\
$N$ & Dimension of each subsystem \\
$\bB(\mathcal{H})$ & Set of linear operators acting on $\mathcal{H}$ \\
$\bD(\mathcal{H})$ & Set of density operators on $\mathcal{H}$ \\
$I_{\cH}$ & Identity operator on $\cH$ \\
$\mathrm{Tr}(\cdot)$ & Trace operation \\
\hline
\multicolumn{2}{l}{\textit{Quantum States and Entanglement}} \\
$\ket{\psi}$ & Pure quantum state vector \\
$\rho, \rho'$ & Density matrices \\
$\bfE(\psi)$ & Entanglement entropy of a pure state $\ket{\psi}$ \\
$\bfS(\rho)$ & Von Neumann entropy of state $\rho$ \\
$\lambda_i$ & Schmidt coefficients of a bipartite pure state \\
$\mathbb{H}_s$ & Set of entanglement-constrained states where $\bfE(\psi) \ge s$ \\
\hline
\multicolumn{2}{l}{\textit{Privacy and Channels}} \\
$\cE$ & Quantum channel (Completely Positive Trace-Preserving map) \\
$\cE^\dagger$ & Adjoint of a quantum channel  \\
$\cE_A, \cE_B$ & Local quantum channels acting on subsystems $A$ and $B$ \\
$\cM_a, \cM_b$ & Local quantum measurements on subsystems $A$ and $B$ \\
$\varepsilon$ & Differential privacy budget \\
$\mathcal{K}_\phi$ & Privacy Hamiltonian induced by the mechanism \\
$J_{\max}, J_{\min}$ & Maximal and minimal privacy energy \\
\hline
\multicolumn{2}{l}{\textit{Others}} \\
$\mathrm{U}_{\cH}$ & Unitary group on Hilbert space $\cH$ \\
$\mathrm{grad} f$ & Riemannian gradient of function $f$ \\
\hline
\end{tabular}
\caption{Summary of notation used in this paper.}
\label{tab:notation}
\end{table}
\section{Technical Lemmas}
In this appendix, we state some technical lemmas for use in the subsequent proofs.

\begin{lemma}[Cauchy interlacing theorem, \cite{bhatia2013matrix}]\label{lemma:Cauchy}
Let $X\in\mathbb{C}^{N\times N}$ be Hermitian with eigenvalues $\sigma_1(X)\ge \cdots \ge \sigma_N(X)\ge 0$ and $P\in\mathbb{C}^{N\times M}$ satisfy $P^\dg P=I_M$.
Then the eigenvalues
\[
\sigma_j(P^\dg X P) \le \sigma_j(X), \qquad j=1,\dots,M.
\]
\end{lemma}

\begin{lemma}[von Neumann trace inequality, \cite{bhatia2013matrix}]\label{lemma:trace-ineq}
Let $X,Y\in\mathbb{C}^{N\times N}$ be positive semidefinite matrices. Then
\[
 \sum_{j=1}^N \sigma_j(X)\,\sigma_{N+1-j}(Y) \le \Tr(X Y) \le \sum_{i=1}^N \sigma_i(X)\,\sigma_i(Y),
\]
where $\sigma_1(\cdot)\ge \cdots \ge \sigma_N(\cdot)\ge 0$ denote eigenvalues arranged in nonincreasing order.
\end{lemma}

\begin{lemma}[Majorization bound for weighted sums]\label{lemma:majorization-dec}
Let $x=(x_1,\dots,x_N)$ and $y=(y_1,\dots,y_N)$ be nonnegative vectors with
$x_1\ge \cdots \ge x_N$ and $y_1\ge \cdots \ge y_N$.
If $z=(z_1,\dots,z_N)$ is another vector satisfying
\[
\sum_{j=1}^M z_j \le \sum_{j=1}^M x_j
\quad \text{for all } M=1,\dots,N,
\]
then
\[
\sum_{j=1}^N y_j z_j \le \sum_{j=1}^N y_j x_j.
\]
\end{lemma}

\begin{proof}
Define the partial sums $S_m:=\sum_{j=1}^m(x_j-z_j)$ for $m=1,\ldots,N$. By assumption, $S_M \ge 0$ for all $M$. Using summation by parts, we write
\begin{align*}
\sum_{j=1}^N y_j (x_j - z_j) = y_NS_N+ \sum_{j=1}^{N-1} (y_j - y_{j+1})S_j,
\end{align*}
Since $y_1 \ge y_2 \ge \cdots \ge y_N \ge 0$, we have $y_j - y_{j+1} \ge 0$ for all $j$. Moreover, by definition of $S_j$, each partial sum satisfies $S_i \ge 0$. Therefore, every term in the above sum is nonnegative, and hence
\[
\sum_{j=1}^N y_j(x_j - z_j) \ge 0.
\]
which is equivalent to $\sum_{j=1}^N y_j z_j \le \sum_{j=1}^N y_j x_j$.
\end{proof}

\section{Proof of Theorem~\ref{thm:ldp-characterization}}\label{app:thm:ldp-characterization}
For any local measurement $\{ \cM_a \otimes \cM_b \}_{(a,b)\in O_A \times O_B}$, any subset \( T \subseteq O_A \times O_B \), and state \(\rho\), we have
\begin{align*}
\Pr\Big[ \cM(\mathcal E(\rho)) \in T \Big] &= \sum_{(a,b)\in T}
\Tr \big( (\cM_a \otimes \cM_b) \mathcal E(\rho) \big) \\
&= \sum_{(a,b)\in T} \Tr\left( (\mathcal E_A^\dagger \otimes \mathcal E_B^\dagger)
(\cM_a \otimes \cM_b) \rho \right).
\end{align*}

Since each POVM \(\cM_a \succeq 0\) is positive semidefinite on \(\cH_A\), it admits a spectral decomposition
\begin{align*}
    \cM_a =\sum_{j=1}^{N} \eta_{a,j} \ket{\phi_{a,j}}\bra{\phi_{a,j}}.
\end{align*}
where \(\{\ket{\phi_{a,j}}\}_{j=1}^{N}\subset \cH_A\) is an orthonormal basis and
\(\eta_{a,j}>0\). Similarly, we have
\begin{align*}
\cM_b =\sum_{k=1}^{N} \eta_{b,k} \ket{\phi_{b,k}}\bra{\phi_{b,k}}, \quad \eta_{b,k}>0.
\end{align*}
Denoting $\ket{\phi_{j,k}} =  \ket{\phi_{a,j}} \otimes \ket{\phi_{b,k}}$, we obtain
\begin{align}\label{eq:app2-sum-tr}
\Pr\big[ M(\mathcal E(\rho)) \in T \big]
&= \sum_{(a,b)\in T}\sum_{j=1}^{N}\sum_{k=1}^{N}\eta_{a,j}\eta_{b,k}
\Tr\Big( (\cE_{A}^\dagger \otimes \cE_{B}^\dagger) 
\bigl( \ket{\phi_{a,j}}\bra{\phi_{a,j}} \otimes \ket{\phi_{b,k}}\bra{\phi_{b,k}} \bigr)
\rho
\Big) \nonumber \\ 
& = \sum_{(a,b)\in T}\sum_{j=1}^{N}\sum_{k=1}^{N}\eta_{a,j}\eta_{b,k}
\Tr\big( \cK_{\phi_{j,k}} \rho \big).
\end{align}

Let $\rho=\ket{\psi}\bra{\psi}$ and $\rho'=\ket{\psi'}\bra{\psi'}$ be any two pure
entangled states in $\bH_s$. Define 
\begin{align*}
    Q(s):= &\max_{\rho,\rho'\in \bH_s}\Bigg\{\Pr\big[ \cM(\mathcal E(\rho)) \in T \big] - e^{\varepsilon} \Pr\big[ \cM(\mathcal E(\rho')) \in T \big] \Bigg\} 
\end{align*}
Using~\eqref{eq:app2-sum-tr} and linearity, we obtain
\begin{align*}
    Q(s) \le  &\sum_{(a,b)\in T}\sum_{j=1}^{N}\sum_{k=1}^{N}\eta_{a,j}\eta_{b,k}
    \max_{\rho, \rho' \in \bH_s}\Big\{\Tr\big( \cK_{\phi_{j,k}} \rho \big) - e^{\varepsilon} \Tr\big( \cK_{\phi_{j,k}} \rho'\big)\Big\} 
    \\
    \le  &\sum_{(a,b)\in T}\sum_{j=1}^{N}\sum_{k=1}^{N}\eta_{a,j}\eta_{b,k} 
    \Big\{\max_{\rho \in \bH_s}\Tr\big( \cK_{\phi_{j,k}} \rho \big) - e^{\varepsilon} \min_{\rho' \in \bH_s} \Tr\big( \cK_{\phi_{j,k}} \rho'\big)\Big\}
    \\
    \le  &\sum_{(a,b)\in T}\sum_{j=1}^{N}\sum_{k=1}^{N}\eta_{a,j}\eta_{b,k} 
    \Big\{\max_{\ket{\psi} \in \bH_s} \bra{\psi}\cK_{\phi_{j,k}} \ket{\psi} - e^{\varepsilon} \min_{\ket{\psi'} \in \bH_s} \bra{\psi'} \cK_{\phi_{j,k}} \ket{\psi'}\Big\}
\end{align*}
Therefore, if
\begin{align*}
    \varepsilon \ge \log \Bigg\{ \max_{\ket{\phi_a} \in \cH_A} \max_{\ket{\phi_b} \in \cH_B} \frac{\max_{\ket{\psi} \in \bH_s} \bra{\psi}\cK_{\phi} \ket{\psi} }{\min_{\ket{\psi} \in \bH_s} \bra{\psi} \cK_{\phi} \ket{\psi}} \Bigg\},
\end{align*}
then each summand is non-positive, i.e.,
\begin{align*}
    \max_{\ket{\psi} \in \bH_s} \bra{\psi}\cK_{\phi_{j,k}} \ket{\psi} - e^{\varepsilon} \min_{\ket{\psi'} \in \bH_s} \bra{\psi'} \cK_{\phi_{j,k}} \ket{\psi'} \le 0,
\end{align*}
and hence $Q(s)\le 0$. This implies an upper bound on $\varepsilon^* (s)$
\begin{align}\label{eq:elm-ub}
    \varepsilon^*(s) \le \log \Bigg\{ \max_{\ket{\phi_a} \in \cH_A} \max_{\ket{\phi_b} \in \cH_B} \frac{J_{\max}(\cK_\phi,s)}{J_{\min}(\cK_\phi,s)} \Bigg\}.
\end{align}

To show tightness, choose
\begin{align*}
    \ket{\phi^*} = \ket{\phi_a^*} \otimes \ket{\phi_b^*} = \arg\max_{\ket{\phi_a} \in \cH_A} \arg\max_{\ket{\phi_b} \in \cH_B} \frac{J_{\max}(\cK_\phi,s)}{J_{\min}(\cK_\phi,s)}
\end{align*}
and let $\ket{\psi}$ and $\ket{\psi'}$ attain $J_{\max}(\cK_{\phi^*},s)$ and $J_{\min}(\cK_{\phi^*},s)$, respectively. Consider the binary-outcome product measurement
\begin{align*}
    &\cM_{a,0} = \ket{\phi_a^*}\bra{\phi_a^*}, \quad \cM_{a,1} = I_{\cH_A} - \ket{\phi_a^*}\bra{\phi_a^*}, \\
    &\cM_{b,0} = \ket{\phi_b^*}\bra{\phi_b^*}, \quad \cM_{b,1} = I_{\cH_B} - \ket{\phi_b^*}\bra{\phi_b^*}, 
\end{align*}
Then we have
\begin{align*}
  \Pr\Big[ M(\mathcal E(\ket{\psi} \bra{\psi})) = (0,0) \Big] - e^{\varepsilon^*} \Pr\Big[ M(\mathcal E(\ket{\psi'} \bra{\psi'})) = (0,0) \Big] = J_{\max}(\cK_{\phi^*},s) - e^{\varepsilon^*}  J_{\min }(\cK_{\phi^*},s),
\end{align*}

Requiring the right-hand side to be non-positive yields
\begin{align}\label{eq:elm-lb}
    \varepsilon^*(s) \ge \log\Bigg(\frac{J_{\max}(\cK_{\phi^*},s)}{J_{\min }(\cK_{\phi^*},s)}\Bigg) = \log \Bigg\{ \max_{\ket{\phi_a} \in \cH_A} \max_{\ket{\phi_b} \in \cH_B} \frac{J_{\max}(\cK_\phi,s)}{J_{\min}(\cK_\phi,s)} \Bigg\},
\end{align}

Combining~\eqref{eq:elm-ub} and~\eqref{eq:elm-lb}, we have
\begin{align*}
     \varepsilon^*(s)  = \log \Bigg\{ \max_{\ket{\phi_a} \in \cH_A} \max_{\ket{\phi_b} \in \cH_B} \frac{J_{\max}(\cK_\phi,s)}{J_{\min}(\cK_\phi,s)} \Bigg\},
\end{align*}
which completes the proof. \hfill \qed

\section{Proof of Theorem~\ref{thm:max-product}}\label{app:thm:max-product}

We split the whole proof into three steps. \\ 
\noindent\textbf{Step 1: Explicit expression of KKT condition.}
Since the Euclidean gradients w.r.t.\ $U_A$ and $U_B$ of $\bfF$ are
\begin{align*}
&\nabla_{U_A} \bfF(\lmd,U_A,U_B) = 2 \Lambda^{1/2}\widetilde{\cK}_B\Lambda^{1/2} U_A \cK_{\phi,a}, \\
&\nabla_{U_B} \bfF(\lmd,U_A,U_B) = 2  \Lambda^{1/2}\widetilde{\cK}_A\Lambda^{1/2} U_B \cK_{\phi,b}, \\
 &\nabla_{\lmd_j}\bfF = \frac{1}{2\sqrt{\lmd_j}} \Big(\widetilde{\cK}_A\Lambda^{1/2}\widetilde{\cK}_B
      + \widetilde{\cK}_B\Lambda^{1/2}\widetilde{\cK}_A\Big)_{jj},
\end{align*}
the Riemannian stationarity conditions~\eqref{eq:ST-U} and~\eqref{eq:ST-lambda-max} are
\begin{align}
\textsc{(ST-$U_A$)} & \qquad\grad_{U_A}\bfF =
2U_A\skew \Big( (U_A^{*})^\dg \Lambda^{1/2}\widetilde{\cK}^*_B\Lambda^{1/2} \UAs\cK_{\phi,a}\Big) = 0, \label{eq:grad-UA-riem}
\\
\textsc{(ST-$U_B$)} & \qquad\grad_{U_B}\bfF=
2U_B\skew \Big( (U_B^{*})^\dg \Lambda^{1/2}\widetilde{\cK}_A^*\Lambda^{1/2} U_B^*\cK_{\phi,b}\Big) = 0. \label{eq:grad-UB-riem} 
\\
\textsc{(ST-$\lambda$)} & \qquad \frac{1}{2\sqrt{\lmd_j^*}} \Bigl(\widetilde{\cK}_A^*\Lambda^{*1/2}\widetilde{\cK}_B^*+ \widetilde{\cK}_B^*\Lambda^{*1/2}\widetilde{\cK}_A^* \Bigr)_{jj} = \nu^*-\eta_j^*+\xi^*(1+\log\lmd_j^*), \quad j=1,\ldots,N. \label{eq:ST-lambda-max-ex}
\end{align}

Since $\UAs$ is unitary, \eqref{eq:grad-UA-riem} is equivalent to
\begin{align}\label{eq:skew-zero}
\skew \Big(\UAsd \Lambda^{1/2}\widetilde{\cK}_B^*\Lambda^{1/2} \UAs \cK_{\phi,a}\Big)=0.
\end{align}
Recall that $\skew(Z)=\tfrac12(Z-Z^\dg)$, hence $\skew(Z)=0$ if and only if $Z=Z^\dg$. Since both $U_A^{\dg} \Lambda^{1/2}\widetilde{\cK}^*_B\Lambda^{1/2} U_A$ and $\cK_{\phi,a}$ are Hermitian, the condition~\eqref{eq:skew-zero} is equivalent to
    \begin{align*}                       
    (\UAs)^{\dg}\Lambda^{1/2}\widetilde{\cK}^*_B\Lambda^{1/2}\UAs\cK_{\phi,a}
    = \Big((\UAs)^{\dg}\Lambda^{1/2}\widetilde{\cK}^*_B\Lambda^{1/2}\UAs \cK_{\phi,a}\Big)^{\dg} 
    = \cK_{\phi,a} (\UAs)^{\dg}\Lambda^{1/2}\widetilde{\cK}^*_B\Lambda^{1/2}\UAs
    \end{align*}
Conjugating by $\UAs$ yields the equivalent condition
\begin{align}\label{eq:comm-UA-tilde}
\Lambda^{1/2}\widetilde{\cK}_B^*\Lambda^{1/2} \widetilde{\cK}_A^* = \widetilde{\cK}_A^* \Lambda^{1/2}\widetilde{\cK}_B^*\Lambda^{1/2}.
\end{align}

Applying the same argument to the $U_B$-stationarity
condition~\eqref{eq:grad-UB-riem} gives
\begin{align}\label{eq:comm-UB-tilde}
\Lambda^{1/2}\widetilde{\cK}_A^*\Lambda^{1/2}
\widetilde{\cK}_B^*
=
\widetilde{\cK}_B^*\Lambda^{1/2}
\widetilde{\cK}_A^*\Lambda^{1/2}.
\end{align}

Now define
\[
X:=\Lambda^{1/2}\widetilde{\cK}_A^*\Lambda^{1/2}\succeq 0,
\qquad
Y:=\widetilde{\cK}_B^*\succeq 0.
\]
Equation~\eqref{eq:comm-UB-tilde} shows that $XY=YX$. Since $X$ and
$Y$ are positive semidefinite and commute, they are simultaneously
diagonalizable.

Back to the optimization problem, we have
\begin{align} \label{eq:maxUAUB} 
J_{\max}(\cK_{\phi},s) = & \max_{\lambda\in \mathbb R^N}
\max_{U_A \in {\rm U}(\cH_A)} \max_{U_B \in {\rm U}(\cH_B)} \bfF(\lambda,U_A,U_B)
\\
= & \max_{\lambda\in \mathbb R^N}
\max_{U_A \in {\rm U}(\cH_A)} \max_{U_B \in {\rm U}(\cH_B)} 
\Tr(XY), \nonumber
\\
&\text{subject to} \quad \lmd_j\ge 0, \quad \sum_{j=1}^N \lmd_j =1, \quad \bfH(\lmd)\ge s. \nonumber
\end{align}

Since $X,Y\succeq 0$ and are simultaneously unitarily diagonalizable at stationarity, the maximization of $\Tr(XY)$ is achieved when their eigenvectors are aligned and the eigenvalues are arranged in the same order. More precisely, for a fixed $\lmd \in \mathbb{R^N}$,
\[
\max_{U_A \in {\rm U}(\cH_A)} \max_{U_B \in {\rm U}(\cH_B)} \Tr \bigl(X Y\bigr) = \sum_{j=1}^N \sigma_j(X)\,\sigma_j(Y),
\]
In our setting, we write $Y=\widetilde{\cK}_B$ and note that $\widetilde{\cK}_B$ can be taken diagonal with eigenvalues sorted in descending order. Hence, the optimal alignment is attained without further rotation on subsystem $B$, and we may choose $U_B^*=I_{\cH_B}$.

\medskip 

\noindent\textbf{Step 2. Characterization of the spectrum of $X$.}
To find $U_A^*$, we need to characterize the spectrum of
\[
X=\Lambda^{1/2}\widetilde{\cK}_A\Lambda^{1/2},
\]
for a fixed $\lmd\in\mathbb R^N$. For any $m=1,\ldots,N$, we have
\begin{align}
\sum_{j=1}^m \sigma_j(X) &= \max_{P^\dg P=I_m} \Tr\bigl(P^\dg X P\bigr) \nonumber
\\
&= \max_{P^\dg P=I_m} \Tr\!\bigl(\widetilde{\cK}_A\,\Lambda^{1/2}PP^\dg\Lambda^{1/2}\bigr).
\label{eq:kyfan-start}
\end{align}
Since both $\widetilde{\cK}_A$ and $\Lambda^{1/2}PP^\dg\Lambda^{1/2}$ are positive semidefinite, Lemma~\ref{lemma:trace-ineq} yields
\begin{align}
\Tr\!\bigl(\widetilde{\cK}_A\,\Lambda^{1/2}PP^\dg\Lambda^{1/2}\bigr)
\le
\sum_{i=1}^N \alpha_i\,\sigma_i\!\bigl(\Lambda^{1/2}PP^\dg\Lambda^{1/2}\bigr).
\end{align}
Since $P$ has rank $m$, the matrix $\Lambda^{1/2}PP^\dg\Lambda^{1/2}$ has rank at most
$m$, and its nonzero eigenvalues coincide with those of $P^\dg\Lambda P$. Therefore,
\begin{align}
\sum_{j=1}^m \sigma_j(X)
&\le
\max_{P^\dg P=I_m}\sum_{i=1}^m \alpha_i\,\sigma_i(P^\dg\Lambda P) \nonumber\\
&\le
\sum_{i=1}^m \alpha_i \max_{P^\dg P=I_m}\sigma_i(P^\dg\Lambda P).
\end{align}
By the Cauchy interlacing theorem (Lemma~\ref{lemma:Cauchy}), we have
\begin{align}\label{eq:kyfan-end}
 \sigma_i(P^\dg\Lambda P)\le \lambda_i,
\qquad 1\le i\le m.
\end{align}
Combining~\eqref{eq:kyfan-start} and~\eqref{eq:kyfan-end} yields the bound
\begin{align}\label{eq:kyfanbound}
\sum_{i=1}^m \sigma_i(X)\le \sum_{i=1}^m \alpha_i\lambda_i,
\qquad m=1,2,\ldots,N.
\end{align}

Applying Lemma~\ref{lemma:majorization-dec} to~\eqref{eq:kyfanbound}, we obtain, for any fixed $\lmd$, that
\[
\max_{U_A\in\mathrm{U}(\cH_A)}\max_{U_B\in\mathrm{U}(\cH_B)} \Tr(XY)
\le
\sum_{j=1}^N \alpha_j\beta_j\lambda_j.
\]
Moreover, equality is attained when $U_A^*=I_{\cH_A}$. Consequently, the original maximization problem reduces to
\begin{align*}
J_{\max}
&=
\max_{\lmd}\Tr\!\bigl(\Lambda^{1/2}\cK_A\Lambda^{1/2}\cK_B\bigr)
=
\max_{\lmd}\sum_{j=1}^N \mu_j\lambda_j, \\
&\text{subject to }\quad
\lambda_j\ge 0,\quad \sum_{j=1}^N\lambda_j=1,\quad \bfH(\lmd)\ge s,
\end{align*}
where $\mu_j=\alpha_j\beta_j$.

Under this formulation, the feasible set defined by the probability simplex together with the entropy constraint is convex, and the objective function is linear in $\lmd$. Therefore, the Karush--Kuhn--Tucker conditions are both necessary and sufficient for optimality.

\noindent\textbf{Step 3. Solving the KKT condition.} Given $U_A^*=I_{\cH_A}$ and $U_B^*=I_{\cH_B}$, the KKT conditions~\eqref{eq:ST-U}--~\eqref{eq:CS} are reduced to
\begin{align}
    & \mu_j - \nu + \eta_j - \xi(1 + \log \ls_j)= 0,        \quad \forall j=1,2,\dots,N  \label{eq:simplex-stationarity-max} \\
    &\eta_j \ls_j = 0, \quad \eta_j \ge 0,                    \quad \forall j=1,2,\dots,N  \label{eq:simplex-slack_eta_max} \\
    &\xi \left( \bfH( {\ls}) - s \right) = 0, \quad \xi \ge 0,  \label{eq:simplex-slack_xi_max} \\
    &\ls_j \ge 0, \quad \sum_{j=1}^N \ls_j = 1,  \quad \bfH(\ls) \ge s, \quad \forall j=1,2,\dots,N 
\end{align}
We analyze the solution based on the activity of the entropy constraint.

\paragraph{Case 3.1: $\xi = 0$ (Inactive Entropy Constraint)}
The stationarity condition \eqref{eq:simplex-stationarity-max} simplifies to:
\begin{align*}
    \mu_j - \nu + \eta_j = 0 \implies \eta_j = \nu - \mu_j.
\end{align*}
For any index $j$ where $\ls_j > 0$, the complementary slackness condition \eqref{eq:simplex-slack_eta_max} implies $\eta_j = 0$. Consequently, for such indices, $\mu_j = \nu$. Since $\eta_j \ge 0$ for all $j$, we must have $\nu - \mu_j \ge 0 \implies \nu \ge \mu_j$. Combining these, if $\ls_j > 0$, then $\mu_j = \nu \ge \mu_j$ for all $j$. Therefore, $\ls_j > 0$ is only permissible if $\mu_j$ attains the maximum value $\mu_1$. Thus, we have
\begin{align*}
    \ls_j = 0 \quad \text{if } \mu_j < \mu_{1}=\sigma_{\max}(\cK_\phi).
\end{align*}
The objective function value becomes:
\begin{align*}
    J_{\max}(\cK_\phi,s) = \sum_{\mu_j = \mu_{1}} \mu_{1} \ls_j = \mu_1.
\end{align*}
The maximum entropy is achieved by the uniform distribution over the highest-energy states, yielding $\bfH_{\max} = \log d_{\max} $. Thus, this solution is valid if and only if the required entropy $s \le \log d_{\max}$.

\paragraph{Case 3.2: $\xi > 0$ (Active Entropy Constraint)}
If $s > \log d_{\max}$, from discussion above we have $\xi > 0$. The stationarity condition \eqref{eq:simplex-stationarity-max} implies
\begin{align*}
    \log \ls_j = \frac{\mu_j - \nu + \eta_j}{\xi} - 1.
\end{align*}
Assume there exists some $k$ such that $\ls_k = 0$. Then $\log \ls_k \to -\infty$. For the right-hand side to equate to $-\infty$, we would require $\eta_k \to -\infty$, which violates $\eta_k \ge 0$.
Therefore, we must have $\ls_j > 0$ for all $j = 1, \dots, N$.

Since $\ls_j>0$ for all $j$, the complementary slackness condition~\eqref{eq:simplex-slack_eta_max} implies
$\eta_j=0$. The stationarity condition then yields
\begin{align*}
    \ls_j
    =\exp\left(\frac{\mu_j-\nu}{\xi}-1\right)
    =\frac{e^{\gamma\mu_j}}{Z(\gamma)},
\end{align*}
where $\gamma=1/\xi$ and $Z(\gamma):=e^{\gamma\nu+1}$. To determine the normalization constant $Z(\gamma)$, we invoke the primal feasibility constraint $\sum_{j=1}^N\ls_j=1$, which gives
\[
\frac{1}{Z(\gamma)}\sum_{j=1}^N e^{\gamma\mu_j}=1
\quad\Longrightarrow\quad
Z(\gamma)=\sum_{j=1}^N e^{\gamma\mu_j}.
\]

Substituting this expression back, we obtain optimal distribution
\begin{align*}
    \ls_j = \frac{e^{\gamma\mu_j}}{Z(\gamma)},
\end{align*}
where multiplier $\gamma$ is uniquely determined by the active entropy
constraint $-\sum_{j=1}^N \ls_j\log\ls_j=s$, namely,
\[
\log Z(\gamma)
-\frac{\gamma}{Z(\gamma)}
\sum_{j=1}^N \mu_j e^{\gamma\mu_j}
=s.
\]

Moreover, since $\ls_j>0$ for all $j$, the optimal distribution must assign positive weight to indices with $\mu_j<\mu_{\max}$.
Consequently, the corresponding weighted average satisfies
\begin{align*}
    J_{\max}(\cK_\phi,s) = \sum_{j=1}^N \mu_j\ls_j < \sum_{j=1}^N \mu_1\ls_j = \mu_1 = \sigma_{\max}(\cK_\phi).
\end{align*}
i.e., the objective value is strictly smaller than the maximal eigenvalue.

Finally, a direct differentiation yields
\begin{align*}
    \frac{\partial J}{\partial \gamma} = & \frac{Z''}{Z}-\Bigl(\frac{Z'}{Z}\Bigr)^2=\sum_{j=1}^N \ls_j(\gamma)\mu_j^2-\Big(\sum_{j=1}^N \ls_j(\gamma)\mu_j\Bigr)^2=\mathrm{Var}_{p(\gamma)}(\mu)>0, \\
    \frac{\partial s}{\partial \gamma} = & \frac{Z'}{Z}-J(\gamma)-\gamma J'(\gamma)=-\gamma \mathrm{Var}_{p(\gamma)}(\mu)<0,
\end{align*}
which implies $\frac{\partial J}{\partial s}<0$ and hence that $J$ is a strictly decreasing function of $s$. This completes the proof. \hfill\qed

\section{Proof of Theorem~\ref{thm:min-product}}\label{app:thm:min-product}

As in the proof of Theorem~\ref{thm:max-product}, we can write
\begin{align}\label{eq:minUAUB}
J_{\min}(\cK_{\phi},s)
&=
\min_{\lmd\in\mathbb R^N}\ \min_{U_A\in{\rm U}(\cH_A)}\ \min_{U_B\in{\rm U}(\cH_B)}
\bfF(\lmd,U_A,U_B) \nonumber\\
&=
\min_{\lmd\in\mathbb R^N}\ \min_{U_A\in{\rm U}(\cH_A)}\ \min_{U_B\in{\rm U}(\cH_B)}
\Tr(XY), \nonumber\\
&\text{subject to}\quad \lmd_j\ge 0,\quad \sum_{j=1}^N \lmd_j=1,\quad \bfH(\lmd)\ge s,
\end{align}
where $X=\Lambda^{1/2}\widetilde{\cK}_A\Lambda^{1/2}$
and $Y=\widetilde{\cK}_B$, with $\Lambda=\diag(\lmd)$.

\paragraph{Relaxed lower bound for $s<\log N$.}
Since $Y=\widetilde{\cK}_B\succeq\beta_N I$ and $X\succeq0$,
\[
\Tr(XY)\ge \beta_N\Tr(X)=\beta_N\Tr(\Lambda\widetilde{\cK}_A).
\]
Let $\lambda_1^\downarrow\ge\cdots\ge\lambda_N^\downarrow$ denote the
decreasing rearrangement of the Schmidt weights.  The von Neumann
trace inequality gives
\[
\min_{U_A}\Tr(\Lambda\widetilde{\cK}_A)
=\sum_{i=1}^N\alpha_i\lambda_{N-i+1}^\downarrow.
\]
Define $p_i:=\lambda_{N-i+1}^\downarrow$.  Relabeling does not change
normalization or entropy, so
\[
\min_{U_A,U_B}\Tr(XY)
\ge\sum_{i=1}^N\alpha_i\beta_N p_i
=\sum_{i=1}^N\widetilde\mu_i p_i,
\qquad \widetilde\mu_i:=\alpha_i\beta_N.
\]
Therefore,
\begin{align*}
J_{\min}(\cK_{\phi},s)
&\ge
\min_{p}\ \sum_{i=1}^N \widetilde\mu_i p_i
=: \widetilde J_{\min}(\cK_{\phi},s),
\\
&\text{subject to}\quad p_i\ge0,\quad
\sum_{i=1}^N p_i=1,\quad \bfH(p)\ge s.
\end{align*}

The optimization defining $\widetilde J_{\min}$ has a linear objective
over a convex feasible set and can be solved by the same KKT analysis
as in Theorem~\ref{thm:max-product}; the Gibbs exponent has the
opposite sign because the objective is minimized.  Since
$\widetilde\mu_1\ge\cdots\ge\widetilde\mu_N$, the resulting optimizer
$p_i\propto e^{-\widetilde\gamma\widetilde\mu_i}$ satisfies
$p_1\le\cdots\le p_N$, consistently with its definition as the reverse
of $\lmd^\downarrow$.

To prove the strict monotonicity in item~(iii), write
\[
p_i(\widetilde\gamma)
=\frac{e^{-\widetilde\gamma\widetilde\mu_i}}
{\widetilde Z(\widetilde\gamma)},
\qquad
\widetilde J(\widetilde\gamma)
=\sum_{i=1}^N p_i(\widetilde\gamma)\widetilde\mu_i.
\]
For $\log d_{\min}<s<\log N$, direct differentiation gives
\begin{align*}
\frac{\partial \widetilde J}{\partial\widetilde\gamma}
&=-\operatorname{Var}_{p(\widetilde\gamma)}(\widetilde\mu)<0,\\
\frac{\partial s}{\partial\widetilde\gamma}
&=-\widetilde\gamma
  \operatorname{Var}_{p(\widetilde\gamma)}(\widetilde\mu)<0.
\end{align*}
Consequently,
\[
\frac{\partial \widetilde J}{\partial s}
=\frac{1}{\widetilde\gamma}>0,
\]
so $\widetilde J_{\min}(\cK_\phi,s)$ is strictly increasing throughout
the high-entanglement regime.  At $s=\log N$, the entropy constraint
forces the uniform distribution, which is the limit
$\widetilde\gamma\to0^+$; continuity therefore extends the strict
inequality to the endpoint. This completes the proof.

\section{Proof of Proposition~\ref{prop:qldp-exp-1}}
According to the spectrum structure
\[
\mathrm{spec}(\cK_a)=\{v(p),v(p),w(p),w(p)\},\qquad
\mathrm{spec}(\cK_b)=\{v(q),v(q),w(q),w(q)\},
\]
we have
\begin{align*}
    \mu =  (v(p)v(q),v(p)v(q),w(p)w(q),w(p)w(q)), \qquad
    \widetilde{\mu} =  w(q)\{v(p),v(p),w(p),w(p)\},
\end{align*}
which further indicates $d_{\max} = d_{\min} = 2$. According to Theorems~\ref{thm:max-product} and~\ref{thm:min-product}, the transition point is $s=\log 2$.
We first consider the case that the entanglement entropy satisfies $s \le \log 2$. In this case, the extremal privacy energies are given by the maximal and minimal eigenvalues of $\cK_\phi$, i.e.,
\[
J_{\max}(s)=v(p)v(q), \qquad \widetilde J_{\min}(s)=w(p)w(q), \quad p,q\in[\tfrac12,1]
\]
Maximizing over $p,q\in[\tfrac12,1]$, the optimal ECLM-QLDP leakage level is given by
\begin{align*}
    \varepsilon^*(s)
    \le \log\!\left( \max_{p,q\in[1/2,1]} \frac{v(p)v(q)}{w(p)w(q)} \right)
    = 2\log 3,
\end{align*}
where the maximum is attained at $p=q=1$. This coincides with the leakage level in the non-entangled case.

We then consider the case that the entanglement entropy satisfies $s>\log 2$. We first analyze the maximal privacy budget $J_{\max}(s,\cK_\phi)$. Since the operator $\cK_\phi$ admits only two distinct eigenvalues, each with degeneracy two, the optimal Gibbs distribution must be uniform within each degenerate eigenspace. Let $\tau\in[\tfrac12,1]$ denote the total weight assigned to the eigenspace corresponding to the maximal eigenvalue $v(p)v(q)$, with the remaining weight $1-\tau$ assigned to the minimal-eigenvalue eigenspace $w(p)w(q)$, the maximal privacy energy takes the form
\begin{align*}
    J_{\max}(s)
    = \tau\, v(p)v(q) + (1-\tau)\, w(p)w(q).
\end{align*}
where $\tau$ is uniquely determined by the entanglement entropy constraint, i.e., $s=\log 2-\tau\log\tau-(1-\tau)\log(1-\tau)$. Similarly, the relaxed minimal privacy energy takes the form
\begin{align*}
    \widetilde J_{\min}(s)
    &= (1-\tau)\, v(p)w(q) + \tau\, w(p)w(q).
\end{align*}
Therefore, the optimal ECLM--QLDP leakage level admits the upper bound
\begin{align*}
    \varepsilon^*(s)
    &\le \log\!\left(
    \max_{p,q\in[1/2,1]}
    \frac{\tau\, v(p)v(q) + (1-\tau)\, w(p)w(q)}
         {(1-\tau)\, v(p)w(q) + \tau\, w(p)w(q)}
    \right)
    = \log\!\left(\frac{8\tau+1}{3-2\tau}\right),
\end{align*}
where the maximum is attained at $p=q=1$.

\end{document}